%% file: 0-Infocom-2025.tex
\documentclass[conference]{IEEEtran}
\IEEEoverridecommandlockouts

\usepackage{array}
\usepackage{amsthm}
\usepackage{amsmath}
\usepackage{amssymb}
\usepackage{amsfonts}
\usepackage{balance}
\usepackage{bbding}
\usepackage{booktabs}
\usepackage{caption}
\usepackage{centernot}
\usepackage{cite}
\usepackage{diagbox}
\usepackage{enumitem}
\usepackage{graphicx}
\usepackage{makecell}
\usepackage{multicol}
\usepackage{multirow}
\usepackage{hyperref}
\usepackage{pifont}
\usepackage{lettrine}
\usepackage[linesnumbered, ruled, vlined]{algorithm2e}
\usepackage{siunitx}
\usepackage{subfigure}
\usepackage{textcomp,mathcomp}
\usepackage{threeparttable}
\usepackage{times}
\usepackage{tikz, xcolor, hyperref}
\usepackage{tcolorbox}
\usepackage[T1]{fontenc}
\usepackage{ulem}
\usepackage{url}

\usetikzlibrary{decorations.pathreplacing, arrows.meta}

\newtheorem{theorem}{Theorem}

\newtheorem{definition}{Definition}
\SetKwInput{KwInput}{Input}                
\SetKwInput{KwOutput}{Output}  

\SetCommentSty{mycommfont}

\begin{document}
\title{AsyncSC: An Asynchronous Sidechain for Multi-Domain Data Exchange in Internet of Things}

\author{\IEEEauthorblockN{Lingxiao Yang, Xuewen Dong\textsuperscript{\Envelope}\thanks{\textsuperscript{\Envelope} Corresponding author: \textit{Xuewen Dong} (xwdong@xidian.edu.cn).}, Zhiguo Wan, Sheng Gao, Wei Tong, Di Lu, Yulong Shen, Xiaojiang Du}
	\IEEEcompsocitemizethanks{
		\IEEEcompsocthanksitem \textit{Lingxiao Yang} and \textit{Xuewen Dong} are with the School of Computer Science and Technology, Xidian University, the Engineering Research Center of Blockchain Technology Application and Evaluation, Ministry of Education, and also with the Shaanxi Key Laboratory of Blockchain and Secure Computing, Xi’an 710071, China. \textit{Zhiguo Wan} is with the Zhejiang Lab, Hangzhou 311121, China. \textit{Sheng Gao} is with the School of Information, Central University of Finance and Economics, Beijing 100081, China. \textit{Wei Tong} is with the School of Information Science and Engineering, Zhejiang Sci-Tech University, Hangzhou 310018, China. \textit{Di Lu} and \textit{Yulong Shen} are with the School of Computer Science and Technology, Xidian University, and also with the Shaanxi Key Laboratory of Network and System Security, Xi’an 710071, China. \textit{Xiaojiang Du} is with the School of Engineering and Science, Stevens Institute of Technology, Hoboken 07030, USA.
		\IEEEcompsocthanksitem This work is supported in part by the National Key R\&D Program of China (No. 2023YFB3107500), the National Natural Science Foundation of China (No. 62220106004, 62232013), the Technology Innovation Leading Program of Shaanxi (No. 2022KXJ-093, 2023KXJ-033), the Innovation Capability Support Program of Shaanxi (No. 2023-CX-TD-02), and the Fundamental Research Funds for the Central Universities (No. ZDRC2202).
	}
}

\maketitle

\begin{abstract}
Sidechain techniques improve blockchain scalability and interoperability, providing decentralized exchange and cross-chain collaboration solutions for Internet of Things (IoT) data across various domains. However, current state-of-the-art (SOTA) schemes for IoT multi-domain data exchange are constrained by the need for synchronous networks, hindering efficient cross-chain interactions in discontinuous networks and leading to suboptimal data exchange. In this paper, we propose AsyncSC, a novel asynchronous sidechain construction. It employs a committee to provide Cross-Blockchain as a Service (C-BaaS) for data exchange in multi-domain IoT. To fulfill the need for asynchronous and efficient data exchange, we combine the ideas of aggregate signatures and verifiable delay functions to devise a novel cryptographic primitive called delayed aggregate signature (DAS), which constructs asynchronous cross-chain proofs (ACPs) that ensure the security of cross-chain interactions. To ensure the consistency of asynchronous transactions, we propose a multilevel buffered transaction pool that guarantees the transaction sequencing. We analyze and prove the security of AsyncSC, simulate an asynchronous communication environment, and conduct a comprehensive evaluation. The results show that AsyncSC outperforms SOTA schemes, improving throughput by an average of 1.21 to 3.96 times, reducing transaction latency by 59.76\% to 83.61\%, and maintaining comparable resource overhead.
\end{abstract}

\input{1-Introduction}

\input{2-Preliminaries}

\input{3-SolutionOverview}
\input{4-Building}
\input{5-Security}
\input{6-PerformanceAnalysis}
\input{10-Conclusion}

\normalem 
\bibliographystyle{IEEEtran}
\balance 
\bibliography{IEEEabrv, cite}

\end{document}

%% file: 1-Introduction.tex
\section{Introduction} \label{intro}
The proliferation of the Internet of Things (IoT) has led to ubiquitous data interactions between smart terminals and sensors \cite{he2022collaborative}. For instance, in mobile crowd-sensing, individuals with smartphones and wearables can share data to achieve comprehensive environmental perception \cite{liu2018survey}. According to Cisco, the number of IoT devices worldwide is expected to reach 500 billion by 2030 \cite{Cisco2023}. Thus, equipping IoT with a secure, efficient, and scalable multi-domain data exchange infrastructure is crucial. However, most current IoT data storage and sharing infrastructures are centralized, requiring IoT organizations across domains to trust a third-party center for data management, resulting in scalability issues and single point of failure risks \cite{liu2023survey}.

The increasing interest in blockchain technology has driven researchers to explore decentralized solutions to enhance IoT capabilities \cite{tong2022blockchain, tong2023ti}. However, due to resource limitations of IoT devices and potential privacy concerns, it is difficult to directly apply public blockchains to IoT scenarios. Permissioned blockchains offer data immutability and multi-party maintenance benefits of public blockchains while providing higher performance, flexibility, and privacy protection \cite{dib2018consortium}. Consequently, some solutions explore using permissioned blockchains to meet the dynamic high-performance requirements of IoT scenarios \cite{feng2021consortium}, but still require cross-chain interaction capabilities between permissioned blockchains \cite{zhang2024time}.

Sidechain provides a secure and efficient cross-chain solution with two modes: (i) \textit{Parent-child mainchain-sidechain}. The traditional sidechains are pegged to the mainchain and achieve scaling by offloading traffic from the mainchain. Sidechain nodes monitor the mainchain states, and the mainchain verifies the cross-chain proofs from the sidechain maintainers for cross-chain interactions \cite{Back2014EnablingBI}. (ii) \textit{Parallel sidechain}. Sidechain and mainchain are independent parallel blockchains, which can act as each other's sidechain and verify mutual cross-chain proofs for two-way cross-chain interactions \cite{kiayias2019proof, gavzi2019proof}. Considering context, the latter is more suitable for data exchange among organizations in different IoT domains. Each organization runs an independent permissioned blockchain to manage IoT data and cross-chain communication without third-party dependency through cross-chain proofs.

\begin{figure}[t]
	\vspace{-0.35cm}
	\centering
	\includegraphics[width=2.8in]{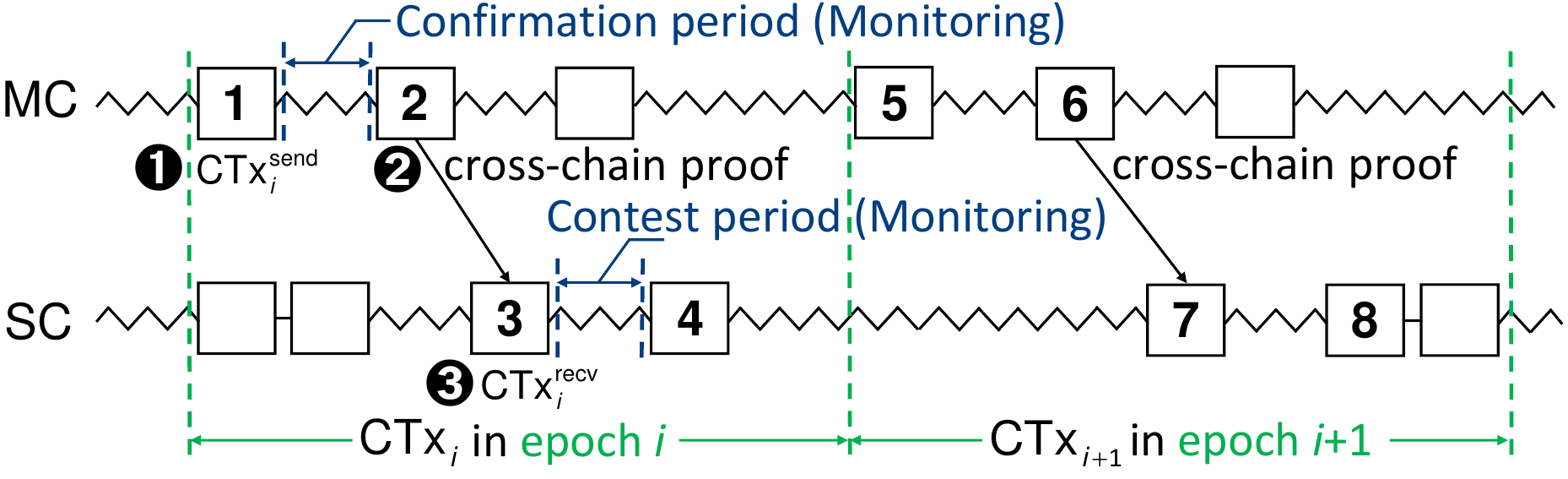}
	\vspace{-0.2cm}
	\captionsetup{font=footnotesize, labelsep=period}
	\caption{Illustration for current sidechain-based cross-chain interactions. Wavy lines indicate omitted blocks. Blocks: \textbf{1.} $\textsf{CTx}_{i}^{\textsf{send}}$ is initiated; \textbf{2.} The cross-chain proof of $\textsf{CTx}_{i}^{\textsf{send}}$ is sent to \textsf{SC}; \textbf{3.} A $\textsf{CTx}_{i}^{\textsf{recv}}$ is created; \textbf{4.} $\textsf{CTx}_{i}^{\textsf{recv}}$ is stabilized on \textsf{SC}, and then $\textsf{CTx}_i$ is completed; \textbf{5 to 8.} The $\textsf{CTx}_{i+1}$ is executed at the next epoch.}
	\label{1-intro}
	\vspace{-0.75cm}
\end{figure}

\textbf{Motivation.} (i) Currently, sidechain applications focus on cryptocurrency exchange on public blockchains, prioritizing security over performance in synchronous network settings \cite{Back2014EnablingBI, kiayias2019proof, gavzi2019proof, yin2021sidechains}, as shown in Fig. \ref{1-intro}. During cross-chain processes, the mainchain and sidechain should continuously monitor the block status, which increases the communication burden (see Section \ref{SBCC} for details). This does not meet the need for efficient cross-domain data exchange in discontinuous networks for resource-constrained devices. (ii) In permissioned blockchains, most schemes use trusted relays for cross-chain interactions \cite{kwon2019cosmos, FISCO, shao2020bitxhub}, but the introduction of third-party relays adds complexity and privacy risk to the system. (iii) Recent advances utilize threshold signatures and zero-knowledge proof techniques to generate succinct cross-chain proofs \cite{yin2023sidechains, garoffolo2020zendoo, xie2022zkbridge}. However, to ensure the consistency of cross-chain interactions, strict synchronization between blockchains is still required, and the continuous monitoring of blocks results in significant on-chain overhead.

In addition, asynchronous cross-chain interaction differs from the synchronous method in that asynchronous cross-chain protocols typically achieve higher performance. This is because it departs from the strict sequential nature of transaction execution in synchronous mode and provides asynchronous concurrency. However, protocol design under asynchronous conditions poses greater challenges and must address how to ensure data consistency and security across different blockchains in the absence of a global clock.

\textbf{Challenge.} Ensuring the persistence and liveness of blockchains for cross-chain interactions without synchronous network settings is a significant challenge for asynchronous cross-chain communication. In other words, it is necessary to ensure that valid asynchronous cross-chain transactions can be smoothly executed, eventually recorded in blocks, and confirmed as stable by a sufficient number of subsequent blocks. This poses the following technical challenges: (i) How to construct a cross-chain interaction mode in asynchronous environments to enable efficient cross-domain data exchange for nodes with resource constraints and unstable network connections. (ii) How to design lightweight and secure asynchronous cross-chain proofs to reduce the computational burden. (iii) How to ensure the sequential consistency of cross-chain transactions in asynchronous mode.

This paper proposes an \textbf{async}hronous \textbf{s}ide\textbf{c}hain alternative construction called AsyncSC, which focuses on a parallel sidechain model without the need for synchronous network settings between blockchains. It utilizes an internal trusted committee formed by blockchain maintainers instead of relying on external relays. The main \textbf{contributions} are as follows:

$\bullet$ \textbf{Asynchronous cross-chain mode.} We present the first asynchronous cross-chain mode suitable for proof-of-stake (PoS) public blockchains and permissioned blockchains, which can construct an honest majority committee. The security of cross-chain interactions is ensured by lightweight asynchronous cryptographic proofs committed by the committee. The committee cyclically provides Cross-Blockchain as a Service (C-BaaS) to enable efficient data interactions across multiple IoT domains in network-unstable environments.

$\bullet$ \textbf{Innovative cryptographic proofs.} We introduce a novel cryptographic primitive called delayed aggregate signature (DAS) to generate asynchronous cross-chain proofs (ACPs) underpinning IoT multi-domain data exchanges. DAS combines the ideas of aggregate signatures and verifiable delay functions (VDFs) to generate an aggregate signature for multiple transactions after a controlled delay and supports public fast verification. As the foundation of C-BaaS, DAS is crucial for secure asynchronous cross-chain interactions, enabling the continuous generation of ACPs in asynchronous mode.

$\bullet$ \textbf{Restricted-readable consistency mechanism.} We design a transaction sequence guarantee mechanism to ensure the consistency of asynchronous cross-chain transactions. It maintains a multilevel buffer pool to coordinate the mainchain and sidechain for transaction ordering and confirmation, thus avoiding conflicts. The mechanism ensures the internal confidentiality of different IoT data domains, making the transaction sequence readable only to the authenticated leader.

$\bullet$ \textbf{Comprehensive evaluation.} We implement a prototype based on Hyperledger Fabric \cite{androulaki2018hyperledger}, and ChainMaker \cite{chainmaker} and the evaluation results demonstrate that AsyncSC outperforms SOTA solutions. It improves transaction throughput by an average of 1.21 $\sim$ 3.96 $\times$, reduces latency by 59.76\% to 83.61\%, enhances success ratio by a relative 46.03\% to 135.9\%, and maintains comparable overhead.


%% file: 2-Preliminaries.tex
\section{Background and Preliminaries} \label{pre}
\subsection{Current Sidechain-Based Cross-Chain Interactions} \label{SBCC}
To achieve better security, existing sidechain-based cross-chain interaction schemes \cite{Back2014EnablingBI, gavzi2019proof, kiayias2019proof, yin2021sidechains, yin2023sidechains} require the mainchain (\textsf{MC}) and sidechain (\textsf{SC}) to perform cross-chain transactions (\textsf{CTx}s) in (semi-)synchronous network settings \cite{david2018ouroboros}. The term ``synchronous setting'' assumes that the maintainers of \textsf{MC} and \textsf{SC} have roughly synchronized clocks to denote the current time slot. Message exchange among transaction participants incurs a maximum delay of $\Delta$ time slots. A scenario with $\Delta=0$ is referred to as a synchronous model, while a scenario with $\Delta>0$ is termed semi-synchronous.

We now proceed to illustrate the example in Fig. \ref{1-intro} under synchronous network mode. Assuming a $\textsf{CTx}_i$ for cryptocurrency exchange from \textsf{MC} to \textsf{SC}, it typically involves three steps: \ding{202} A $\textsf{CTx}_{i}^{\textsf{send}}$ locks coins on \textsf{MC}. \ding{203} \textsf{MC} generates a cross-chain proof committing to the execution of $\textsf{CTx}_{i}^{\textsf{send}}$ after monitoring the stabilization of the block recording $\textsf{CTx}_{i}^{\textsf{send}}$ on the ledger $\textsf{L}_{\textsf{MC}}$ through a sufficient number of subsequent block confirmations\footnote{In the Bitcoin blockchain, it typically takes six subsequent blocks to confirm the stability of a transaction within a block.} (also called the confirmation period \cite{Back2014EnablingBI}). \ding{204} \textsf{SC} receives the cross-chain proof under the synchronous network, verifies it and executes the $\textsf{CTx}_{i}^{\textsf{recv}}$ to create the corresponding coins. Similarly, $\textsf{CTx}_i$ is not successfully executed until the block containing $\textsf{CTx}_{i}^{\textsf{recv}}$ is stabilized on the $\textsf{L}_{\textsf{SC}}$ (\textit{i.e.}, after the contest period). The subsequent $\textsf{CTx}_{i+1}$ is not executed until the above steps of $\textsf{CTx}_i$ are completed.

The primary performance bottleneck in the existing cross-chain interaction process arises from the continuous monitoring required by \textsf{MC} and \textsf{SC} to determine the stability of $\textsf{CTx}_{i}^{\textsf{send}}$ and $\textsf{CTx}_{i}^{\textsf{recv}}$, respectively, before initiating a new $\textsf{CTx}_{i +1}$. The \textsf{CTx}s are sequentially executed under the synchronous network. Furthermore, the monitoring processes of \textsf{MC} and \textsf{SC} necessitate constant state synchronization, thus increasing the communication load. In addition, we compare the features of AsyncSC with existing sidechains in Table \ref{features}.

\begin{table}[t]
	\vspace{-0.15cm}
	\captionsetup{font=footnotesize}
	\caption{\textsc{Feature Comparison with Existing Sidechains.}}
	\vspace{-0.2cm}
	\hspace{-0.32cm}
	\resizebox{1.02\columnwidth}{!}{
		\begin{threeparttable}
			\begin{tabular}{lccccc}
				\toprule
				\diagbox{Schemes}{Features} & \makecell{Sidechain\\ mode} & \makecell{Cross-chain\\ proof$^1$} & \makecell{Asynchronous\\ network/mode} & \makecell{Batch\\ commitment$^2$}  & \makecell{IoT data\\ exchange$^3$} \\
				\midrule
				Pegged sidechains \cite{Back2014EnablingBI} & Parent-child & SPV proof  & \ding{55} & \ding{55} & \ding{55}  \\
				PoW sidechains \cite{kiayias2019proof} & Parallel & NIPoPoWs & \ding{55} & \ding{55} & \ding{55}  \\
				PoS sidechains \cite{gavzi2019proof} & Parallel & ATMS  & \ding{55} & \ding{55} & \ding{55}  \\
				Ref. \cite{yin2021sidechains}, Ge-Co \cite{yin2023sidechains} & Parent-child  & TSS  & \ding{55} & \ding{55} & \ding{55} \\
				AsyncSC (This work) & Parallel & ACP & \ding{51} & \ding{51} & \ding{51} \\
				\bottomrule
			\end{tabular}
			\begin{tablenotes}
				\item[$1$] SPV stands for simplified payment verification; NIPoPoWs stands for non-interactive proofs of proof-of-work; ATMS stands for ad-hoc threshold multi-signatures; TSS stands for threshold signature schemes.
				\item[$2$] Each cross-chain proof in existing sidechains supports committing only a single transaction at a time.
				\item[$3$] Existing sidechains only support cross-chain exchanges involving cryptocurrencies.
			\end{tablenotes}
		\end{threeparttable}
	}
	\label{features}
	\vspace{-0.6cm}
\end{table}

\subsection{Cross-Chain Interactions in Asynchronous Networks} \label{syn}
To meet the needs of resource-constrained multi-domain IoT devices for cross-chain data exchange in discontinuous networks, IoT devices must be able to store data locally while offline and upload it to the blockchain once online. This aspect is not the primary focus of this paper, as existing implementations are already mature \cite{muzammal2019renovating, lu2021blockchain, zhou2022fair}. More importantly, cross-chain interactions between blockchains need to support asynchronous transmissions. 

To distinguish our proposal from existing work, we elucidate cross-chain interactions in asynchronous networks. As shown in Fig. \ref{2-async}, under synchronous networks, cross-chain proofs are assumed to be transmitted from \textsf{MC} to \textsf{SC} within a finite time, and the timing required for each \textsf{CTx} does not overlap. However, in this paper's context, message transmission between \textsf{MC} and \textsf{SC} may be interrupted, causing \textsf{CTx}s in synchronous mode to be blocked. To ensure proper business operations, we propose allowing the timings of \textsf{CTx}s to partially overlap in asynchronous mode.

Specifically, after $\textsf{CTx}_{i}^{\textsf{send}}$ has been executed (\textit{i.e.}, the transaction is signed, consensused, and recorded on-chain) on \textsf{MC}, the generation of cross-chain proof can begin. The cross-chain proof of \textsf{MC} is delayed to be output after $\textsf{CTx}_{i}^{\textsf{send}}$ is stabilized, and this delayed output can be verified by \textsf{SC}. In other words, this delayed cross-chain proof commits the stability of $\textsf{CTx}_{i}^{\textsf{send}}$ in \textsf{MC}, so \textsf{MC} does not need to continuously monitor the state of $\textsf{CTx}_{i}^{\textsf{send}}$. This reduces communication overhead and allows \textsf{MC} to concurrently execute the next $\textsf{CTx}_{i+1}$, thereby improving transaction efficiency.

In an asynchronous network, \textsf{MC}'s cross-chain proofs may be transmitted to \textsf{SC} after unknown delays due to unstable communications. Thus, \textsf{SC} needs to ensure that the order of asynchronous transactions is consistent with \textsf{MC}. We will explain this issue further in Section \ref{Problem}.

\subsection{Cryptographic Primitives} \label{CP}
\textit{\textbf{Aggregate Signature.}} The aggregate signature is a cryptographic digital signature technique that supports aggregation \cite{boneh2001short}. Given $n$ signatures from $n$ different users on $n$ different messages, it can aggregate these signatures into a single, short signature. This aggregate signature convinces the verifier that $n$ users indeed signed the $n$ original messages (\textit{i.e.}, user $i$ signed message $M_i$ for $i=1,...,n$). This reduces storage, transmission, and verification costs. In blockchain, aggregate signatures enable batch verification of transactions \cite{zhao2019practical}.

\textbf{\textit{Verifiable Delay Function.}} The verifiable delay function (VDF) is a cryptographic function \cite{boneh2018verifiable}. It allows a prover to use the public parameters $pp$ to generate a deterministic, unique value $y$ corresponding to the input $x$ and a proof $\pi$ by running an evaluation algorithm $(y, \pi) \leftarrow \texttt{Eval}(pp, x)$ with input arbitrary data $x$. The execution of \texttt{Eval} requires a specified number of consecutive steps, causing a delay predetermined by a delay parameter $t$. Once \texttt{Eval} is computed, anyone can use the public parameters $pp$ to quickly and publicly verify the output of \texttt{Eval} via $\{accept, reject\} \leftarrow \texttt{Verify}(pp, x, y, \pi)$.

\begin{figure}[t]
	\vspace{-0.3cm}
	\centering
	\includegraphics[width=3.3in]{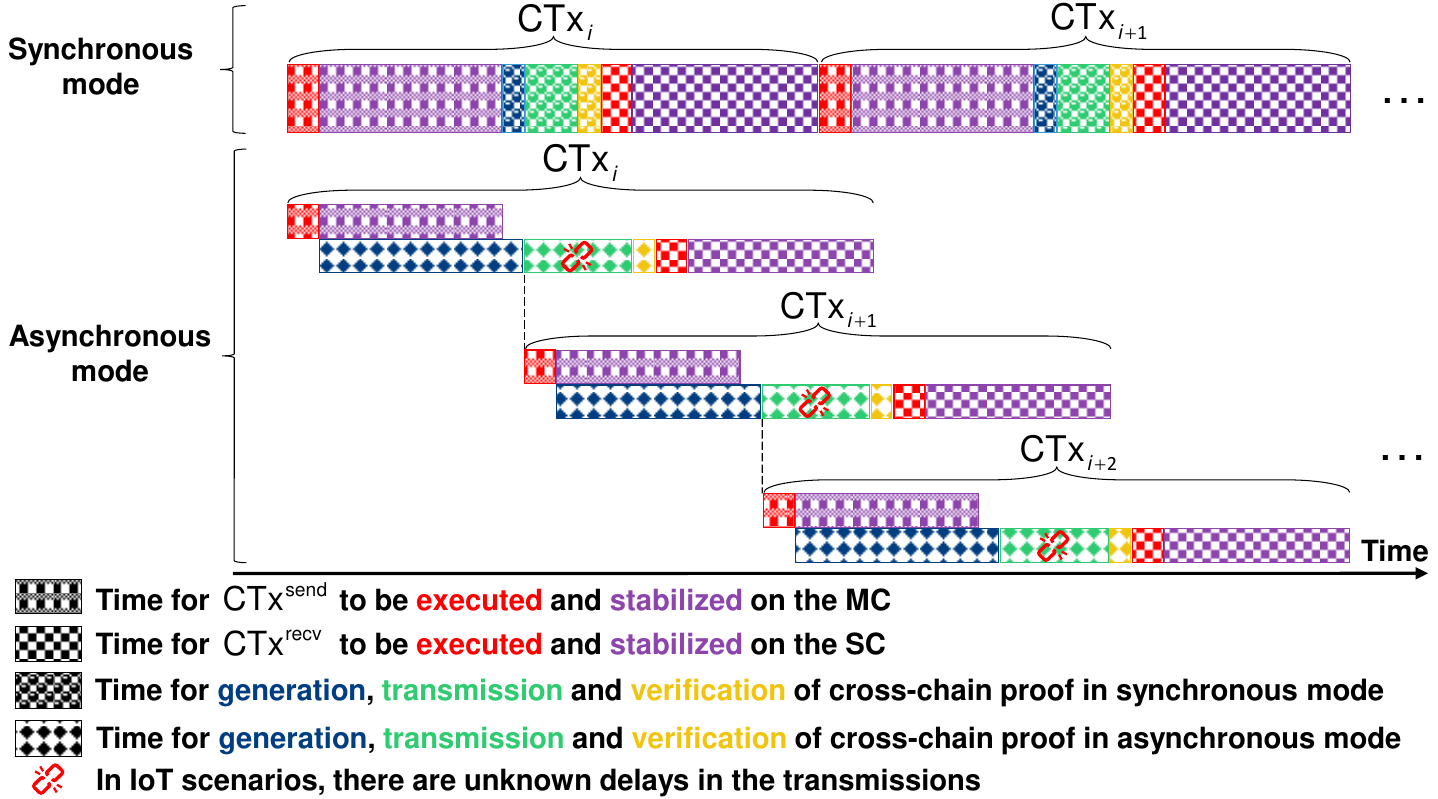}
	\vspace{-0.1cm}
	\captionsetup{font=footnotesize, labelsep=period}
	\caption{Illustration for asynchronous cross-chain interactions.}
	\label{2-async}
	\vspace{-0.65cm}
\end{figure}

%% file: 3-SolutionOverview.tex
\section{AsyncSC: System Overview} \label{solu}
This section introduces the AsyncSC system model, discusses the transaction consistency problem in asynchronous cross-chain mode, and presents the assumptions. 

\subsection{System Model} \label{System}
This paper focuses on multi-domain data exchange scenarios in IoT, as illustrated in Fig. \ref{3-system}. IoT sensor data is uploaded to permissioned blockchains managed by different organizations via transport nodes like IoT gateways and base stations for decentralized management. Independent blockchains in a parallel chain relationship require cross-chain interaction for data exchange among different organizations. Due to resource constraints of IoT devices and unstable network connections, cross-chain communication is achieved by generating and verifying asynchronous cross-chain proofs (ACPs).

Next, we describe the entities involved in cross-chain interactions and the workflow depicted in Fig. \ref{3-entity}.

\textbf{Entities.} AsyncSC includes the following three entities.

\textbf{\textit{Organization.}} It consists of permissioned blockchain nodes $S^{Org}$ that maintain data from different IoT domains. These nodes may be elected to a committee to provide Cross-Blockchain as a Service (C-BaaS), which records data from IoT devices on-chain and facilitates cross-chain interactions.

\textbf{\textit{Committee.}} It is a set of nodes $S^{C}=\{C_{i}\}|_{i=1}^{m} \in S^{Org}$ that are fairly elected from the organization by periodically running an election protocol\footnote{Existing committee election schemes are well established, as noted in the Ref. \cite{yin2023sidechains, huang2023scheduling, zhai2024secret, boehmer2024approval}, and their optimization is part of our future work.} via smart contracts. These nodes provide C-BaaS, verify and sign transactions $\{\textsf{CTx}_i\}|_{i=1}^{n}$ packed per epoch, and generate ACPs by the leader. Additionally, the committee maintains a buffer pool to ensure the consistency of asynchronous \textsf{CTx}s.

\textbf{\textit{Leader.}} It is a node $C^L$ with the highest election value in the committee, \textit{i.e.}, $C^L \in S^{C}$ and $C^L \leftarrow \texttt{max.ElecVal}(\{C_{i}\}|_{i=1}^{m})$. It is responsible for performing delayed signature aggregation of the \textsf{CTx} set $\textsf{CTxSet}:=\{\textsf{CTx}_i\}|_{i=1}^{n}$ for each epoch of the committee and generating ACPs. Additionally, it verifies the received ACPs and finalizes the corresponding \textsf{CTx}s.

\begin{figure}[t]
	\vspace{-0.4cm}
	\centering
	\includegraphics[width=3.2in]{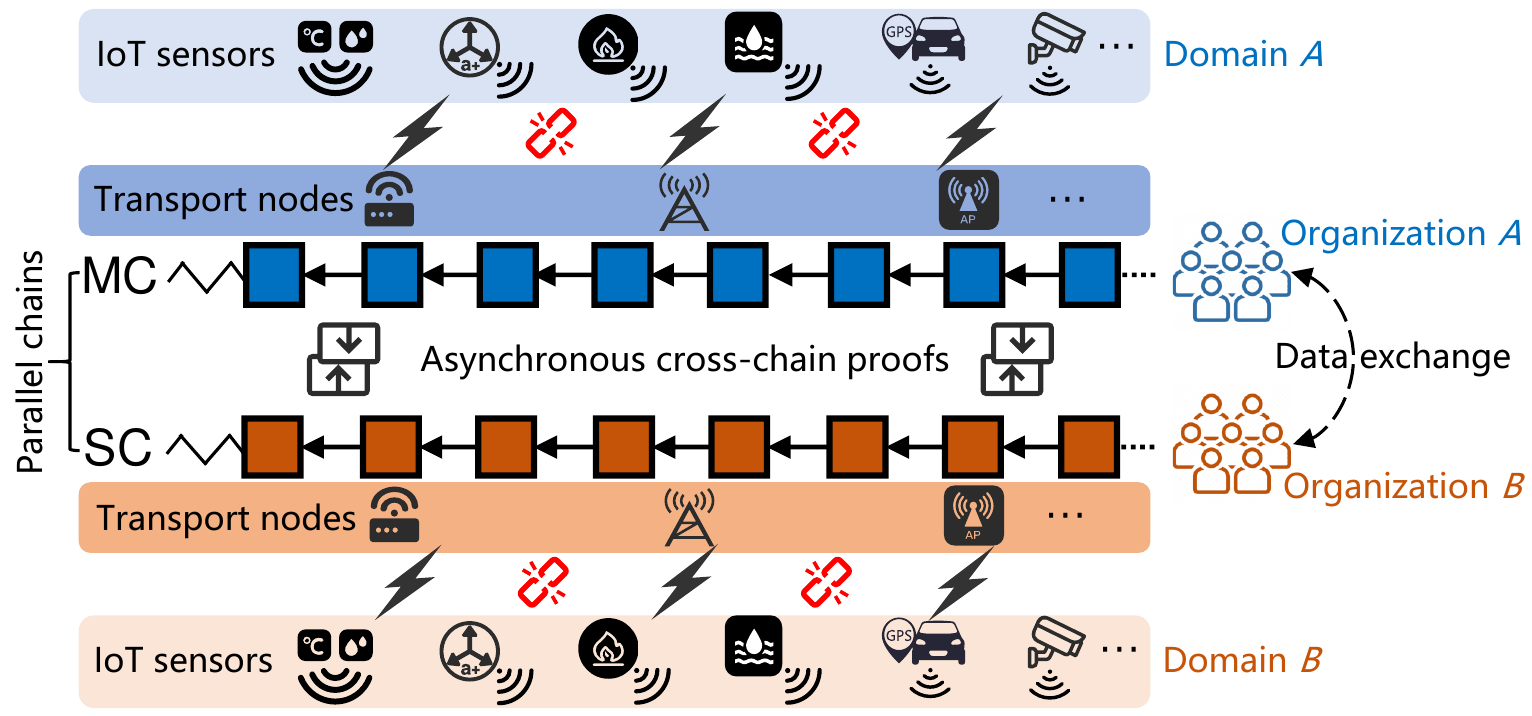}
	\vspace{-0.1cm}
	\captionsetup{font=footnotesize, labelsep=period}
	\caption{System model of AsyncSC.}
	\label{3-system}
	\vspace{-0.71cm}
\end{figure}

\textbf{Workflow.} The C-BaaS consists of the following steps:

\ding{202} \textbf{\textit{\textsf{CTx} initiation.}} Suppose in an epoch, the organization to which \textsf{MC} belongs initiates a batch of \textsf{CTx}s, \textit{i.e.}, $\{\textsf{CTx}_i\}|_{i=1}^{n}$. The nodes of \textsf{MC} start to record $\{\textsf{CTx}_{i}^{\textsf{send}}\}|_{i=1}^{n}$ on-chain. At the same time, the committee of \textsf{MC} receives a \textsf{CTxSet} message formed by packing these transactions.

\ding{203} \textbf{\textit{Committee signing.}} The committee of \textsf{MC} starts to provide C-BaaS. First, the committee puts the transactions of \textsf{CTxSet} into the buffer pool. Then, each committee member $C_{i} \in S^{C}$ signs the \textsf{CTxSet} using his/her private key via any EUF-CMA-secure digital signature scheme, \textit{i.e.}, $\sigma_{i} \leftarrow \texttt{Sig}(\textsf{CTxSet}, sk_i)$ for $i=1,... ,m$. Finally, the committee sends a sequence of $\{(pk_i, \sigma_i)\}|_{i=1}^{m}$ to the leader of \textsf{MC}.

\ding{204} \textit{\textbf{ACP generation.}} The leader of \textsf{MC} executes a delayed signature aggregation on the committee's signatures, \textit{i.e.}, $(\sigma_{\textsf{DAS}}, \pi_{\textsf{DAS}}, \cdot) \leftarrow \texttt{DASig}(\cdot)$. Here, the $\cdot$ notation indicates that some inputs or outputs of the algorithm are omitted (details in Section \ref{Definition-DAS}). The outputs of \texttt{DASig} along with the public parameter $pp$ form an ACP representing the validity commitment of \textsf{CTxSet} on \textsf{MC}.

\ding{205} \textit{\textbf{ACP verification.}} Upon receiving the ACP from \textsf{MC}, the leader of \textsf{SC} verifies the validity of the ACP, \textit{i.e.}, $1/0 \leftarrow \texttt{DASVer}(\cdot)$. Upon successful verification, the committee of \textsf{SC} queues the transactions within \textsf{CTxSet} into the buffer pool awaiting on-chain recording. To ensure the consistency of asynchronous \textsf{CTx}s, the \textsf{SC} committee conducts a transaction sequentiality check utilizing the buffer pool maintained by \textsf{MC} to prevent conflicts.

\ding{206} \textbf{\textit{\textsf{CTx} completion.}} The organization to which \textsf{SC} belongs records the $\{\textsf{CTx}_{i}^{\textsf{recv}}\}|_{i=1}^{n}$ in \textsf{CTxSet} on-chain. Once these transactions are stabilized within the ledger $\textsf{L}_{\textsf{SC}}$, the corresponding \textsf{CTx}s for this epoch are finalized.

\subsection{Problem Statement} \label{Problem}
Due to unstable network connections across various IoT domains, asynchronous communication occurs between blockchains. For instance, during epoch $i$, the \textsf{SC} may not promptly receive the ACP$_{i}$ generated by \textsf{MC} for that epoch, leading to potential delays in the workflows \ding{205} and \ding{206} as depicted in Fig. \ref{3-entity}. Subsequent transactions on \textsf{SC} may conflict if they depend on transactions committed by ACP$_{i}$.

Specifically, suppose at epoch $i$, \textsf{MC} records a temperature sensor data as 10 $\tccentigrade$. At epoch $i+1$, the ambient temperature rises by 5 $\tccentigrade$, updating the data to 15 $\tccentigrade$. If \textsf{SC} does not receive the temperature data from epoch $i$ in time due to network delay or other reasons, it may record the data at epoch $i+1$ directly as 15 $\tccentigrade$ and ignore the 10 $\tccentigrade$ that should have been recorded. This results in the data being recorded out of order, thereby affecting the consistency and accuracy of the system.

To avoid this issue, we address the consistency problem of asynchronous \textsf{CTx}s by maintaining transaction buffer pools via committees and devising an asynchronous transaction sequence guarantee mechanism (see Section \ref{Buffer}).

\begin{figure}[t]
	\vspace{-0.43cm}
	\centering
	\includegraphics[width=3.1in]{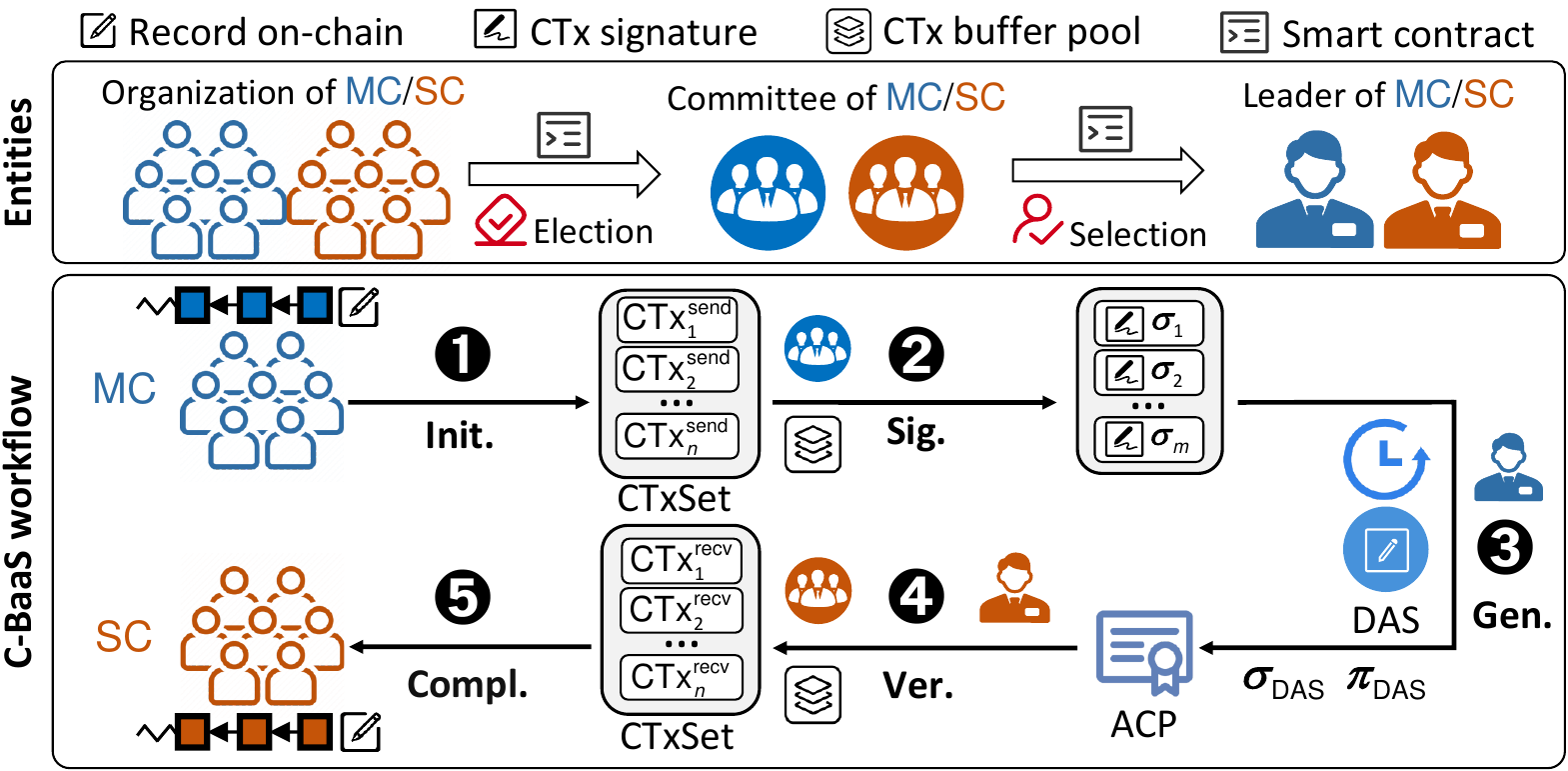}
	\vspace{-0.1cm}
	\captionsetup{font=footnotesize, labelsep=period}
	\caption{Cross-chain interaction entities and workflow in AsyncSC.}
	\label{3-entity}
	\vspace{-0.65cm}
\end{figure}

\subsection{Assumptions} \label{Assumptions}
\textbf{\textit{Assumptions.}} The upper bound $\Delta$ on message transmission delay due to asynchronous communication may be unknown. In practice, we assume a negligible probability of $\Delta=\infty$ between \textsf{MC} and \textsf{SC}. To ensure safety and liveness, \textsf{MC} retransmits the \textsf{CTx}s to \textsf{SC} if they exceed an asynchronous response time threshold $\Delta_{async}$. The \textsf{CTx}s of each epoch eventually complete their execution through the C-BaaS provided by the committee. If the verification of \textsf{MC}'s cross-chain proof passes, then \textsf{SC} will accurately complete the corresponding cross-domain data on-chain and stabilize it. To minimize the overhead caused by the committee elections, we assume a long-term election period, with each period containing multiple epochs for packing transactions.

\textbf{\textit{Threat model.}} We consider that an adversary may attempt to violate the protocol by forging signatures or delaying/denying message transmission. However, the adversary's behavior is rational, \textit{i.e.}, it weighs the costs and benefits of the attack. As in existing blockchain systems, we assume that the nodes maintaining the permissioned blockchains are majority honest \cite{yin2021sidechains, castro1999practical}. The adversary does not possess enough computational power to break the participants' signatures or compromise blockchains and smart contracts.

%% file: 4-Building.tex
\begin{table*}[t]
	\vspace{-0.4cm}
	\centering
	\renewcommand{\arraystretch}{1.2} 
	\resizebox{1.75\columnwidth}{!}{
		\begin{tabular}{|r l l|}
			\hline
			& \multicolumn{1}{c}{\multirow{1}{*}{\centering \begin{tabular}{c}
						$S^{Org}_{A}(S^{C}_{A}, C^L_A, (PK_A, SK_A), (ek_A, vk_A), \{d_i^A\}|_{i=1}^{n})$
			\end{tabular}}}
			& \multicolumn{1}{c|}{\multirow{1}{*}{\centering \begin{tabular}{c}
						$S^{Org}_{B}(S^{C}_{B}, C^L_B, (PK_B, SK_B), (ek_B, vk_B), \{d_i^B\}|_{i=1}^{n'})$
			\end{tabular}}} \\
			\hline
			1: & \multicolumn{2}{c|}{\begin{tikzpicture}
					\draw[very thin,->] (0,0) -- (8.5,0) node[midway, above] {Cross-domain transfer of multiple data $\{d_i^A\}|_{i=1}^{n}$ from $S^{Org}_{A}$ to $S^{Org}_{B}$};
			\end{tikzpicture}} \\ 
			2: & Include each data $d_i^A$ within an epoch into each $\textsf{CTx}_i$ of $\textsf{CTxSet}_A$ & \\
			3: & $\sigma_{i}^A \leftarrow \texttt{Sig}(\textsf{CTxSet}_A, sk_i^A)$ for $i=1,... ,m$ & \\
			4: & $apk_A \leftarrow \texttt{KeyAgg}(PK_A)$ & \\
			5: & $(\sigma_{\textsf{DAS}}, \pi_{\textsf{DAS}}, \textsf{CTxSet}_A') \leftarrow \texttt{DASig}(\textsf{CTxSet}_A, PK_A, \{(pk_i^A, \sigma_i^A)\}|_{i=1}^{m}, ek_A)$ & \\
			6: & Abort if $\texttt{AKCheck}(PK_A, apk_A) \neq 1$ & \\
			7: & \multicolumn{2}{c|}{\begin{tikzpicture}
					\draw[very thin,->] (0,0) -- (7.2,0) node[midway, above] {$ACP_A:= (vk_A, \textsf{CTxSet}_A, apk_A, \sigma_{\textsf{DAS}}, \pi_{\textsf{DAS}}, \textsf{CTxSet}_A')$};
			\end{tikzpicture}} \\ 
			8: & & Abort if $\texttt{DASVer}(\textsf{CTxSet}_A, apk_A, \sigma_{\textsf{DAS}}, \pi_{\textsf{DAS}}, \textsf{CTxSet}_A', vk_A) \neq 1$ \\
			9: & & $\{\textsf{CTx}_i\}|_{i=1}^{n} \leftarrow \textsf{CTxSet}_A$ \\
			10: & & \textsf{CTx} sequential check and record each cross-domain data $d_i^A$ on-chain \\
			11: & \multicolumn{2}{c|}{\begin{tikzpicture}
					\draw[very thin,<-] (0,0) -- (10,0) node[midway, above] {Cross-domain transfer of multiple data $\{d_i^B\}|_{i=1}^{n'}$ from $S^{Org}_{B}$ to $S^{Org}_{A}$ (vice versa)};
			\end{tikzpicture}} \\
			\hline
		\end{tabular}
	}
	\vspace{-0.1cm}
	\captionsetup{type=figure, font=footnotesize, labelsep=period, justification=justified, singlelinecheck=false}
	\caption{Procedure of applying DAS to AsyncSC. In the figure, organization \textit{A} cross-domain transfers data $\{d_i\}|_{i=1}^{n}$ within an epoch to organization \textit{B} and vice versa. $(PK_A, SK_A)$ represents the public-private key pairs of the committee nodes $S^{C}_{A} \in S^{Org}_{A}$ in organization \textit{A}. $(ek_A, vk_A)$ are the publicly accessible evaluation and verification keys used in DAS by organization \textit{A}.}
	\label{DAS-procedure}
	\vspace{-0.6cm}
\end{table*}

\section{Delayed Aggregate Signature} \label{DAS}
In this section, we introduce a novel cryptographic primitive known as the delayed aggregate signature (DAS), which serves as the fundamental building block of ACP and the core of C-BaaS. DAS integrates properties from aggregate signatures and VDFs. It provides a way to: (i) Aggregate multiple individual signatures $\{\sigma_{i}\}|_{i=1}^{m}$ into a single aggregate signature $\sigma_{\textsf{DAS}}$ after a controllable delay $\delta$ and produce a proof of delay $\pi_{\textsf{DAS}}$. (ii) Enable the verifier to quickly and publicly verify the $\sigma_{\textsf{DAS}}$ and $\pi_{\textsf{DAS}}$, thereby effectively verifying the authenticity of $\{\sigma_{i}\}|_{i=1}^{m}$ and the delay $\delta$.

\subsection{Definition of Delayed Aggregate Signature} \label{Definition-DAS}
\begin{definition}
	{\rm \textbf{(Delayed Aggregate Signature).}} A delayed aggregate signature is a tuple of algorithms $\Pi=$ {\rm$(\texttt{ParGen}, \texttt{KeyGen}, \texttt{Sig}, \texttt{Ver}, \texttt{KeyAgg}, \texttt{AKCheck}, \texttt{DASig},$} {\rm$\texttt{DASVer})$}, whose syntax is described below:
	
	- \underline{$(sp,pp)\leftarrow${\rm$\texttt{ParGen}$}$(1^\lambda, t):$} A parameter generation algorithm that takes as input the security parameter $\lambda$ and the delay parameter $t$, and outputs the global system parameters $sp$ and the public parameters $pp=(ek, vk)$ consisting of an evaluation key $ek$ and a verification key $vk$.
	
	- \underline{$(pk_i, sk_i)\leftarrow${\rm$\texttt{KeyGen}$}$(sp):$} A key generation algorithm that inputs $sp$ and generates a public-private key pair $(pk_i, sk_i)$ for the invoker.
	
	- \underline{$\sigma_i \leftarrow${\rm$\texttt{Sig}$}$(sk_i, M):$} A common signature algorithm that inputs the private key $sk_i$ and a message $M \in \{0, 1\}^*$, and outputs a signature $\sigma_i$.
	
	- \underline{$1/0 \leftarrow${\rm$\texttt{Ver}$}$(M, pk_i, \sigma_i):$} A signature verification algorithm that inputs the message $M$, the public key $pk_i$ and the signature $\sigma_i$, and outputs 1 (accept) if the verification is successful and 0 (reject) otherwise.
	
	- \underline{$apk \leftarrow${\rm$\texttt{KeyAgg}$}$(PK):$} A key aggregation deterministic algorithm that inputs a set of public keys $PK=\{pk_{i}\}|_{i=1}^{m}$ and outputs an aggregate public key $apk$.
	
	- \underline{$1/0 \leftarrow${\rm$\texttt{AKCheck}$}$(PK, apk):$} An aggregate key checking algorithm that inputs a public key set $PK$ and an aggregate public key $apk$, and outputs 1 if the recomputation check is correct and 0 otherwise.
	
	- \underline{$(\sigma_{\textsf{DAS}}, \pi_{\textsf{DAS}}, M') \leftarrow${\rm$\texttt{DASig}$}$(M, PK, \{(pk_i, \sigma_i)\}|_{i=1}^{m}, ek):$} A delayed aggregation signature algorithm that inputs the message $M$, the set of public keys $PK$, a sequence of public key-signature pairs $\{(pk_i, \sigma_i)\}|_{i=1}^{m}$ and the evaluation key $ek$, and outputs a delayed aggregate signature $\sigma_{\textsf{DAS}}$, a proof of delay $\pi_{\textsf{DAS}}$ and a unique output $M'$ corresponding to message $M$.
	
	- \underline{$1/0 \leftarrow${\rm$\texttt{DASVer}$}$(M, apk, \sigma_{\textsf{DAS}}, \pi_{\textsf{DAS}}, M', vk):$} A delayed aggregate signature verification algorithm that inputs the message $M$, the aggregate public key $apk$, the delayed aggregate signature $\sigma_{\textsf{DAS}}$, the delayed proof $\pi_{\textsf{DAS}}$, the unique output $M'$ of the message $M$ and the verification key $vk$, and outputs 1 if the verification is successful and 0 otherwise.
\end{definition}

\subsection{Realization of Delayed Aggregate Signature} \label{Realization-DAS}
The specific realization of DAS is partially inspired by two primitives: aggregate signatures and VDFs, both initially proposed by Boneh \textit{et al.} \cite{boneh2001short, boneh2018verifiable}. Aggregate signatures or multi-signatures have been improved by Ref. \cite{boneh2003aggregate, zhao2019practical, nick2021musig2}, while VDFs have been enhanced by Ref. \cite{wesolowski2020efficient, ephraim2020continuous}. It is noteworthy that the improvements on these primitives are independent of this paper. We extract the signature compression property of aggregate signatures and the delay verifiability property of VDFs to implement compact ACPs.

We refer to constructions in the multi-signature scheme \cite{boldyreva2002efficient} and the VDF scheme for incremental verifiable computation (IVC) \cite{boneh2018verifiable} to realize the DAS scheme. It relies on Gap Diffie-Hellman (GDH) groups and succinct non-interactive arguments of knowledge (SNARK) for its efficient construction. Due to the page limit of this submission, the detailed construction of the DAS, along with its security definition and proof, can be found in the full version published online.

\subsection{Applying DAS to AsyncSC}
To provide the C-BaaS to the organizations managing different IoT domains, the committee $S^{C}$ of each organization $S^{Org}$ runs \texttt{ParGen} and \texttt{KeyGen} to obtain the parameters $(sp,pp)$ and assign the $(pk_i, sk_i)$ pairs.

Fig. \ref{DAS-procedure} illustrates the procedure for the cross-domain transfer of packaged data set $\{d_i^A\}|_{i=1}^{n}$ within an epoch from organization \textit{A} to organization \textit{B}. First, the nodes $S^{Org}_{A}$ of organization \textit{A} initiate \textsf{CTx}s. Each data $d_i^A$ within the epoch is included in its corresponding $\textsf{CTx}_i$, and these \textsf{CTx}s form a $\textsf{CTxSet}_A$ (Line 2). The committee nodes $S^{C}_{A}$ then sign $\textsf{CTxSet}_A$ (Line 3). Subsequently, the leader $C^L_A$ aggregates the public keys $PK_A$ of $S^{C}_{A}$ with $\texttt{KeyAgg}$ (Line 4).

The main procedures then include ACP generation and verification. The fundamental concept underlying asynchronous C-BaaS is that the delay $\delta$ introduced by executing $\texttt{DASig}$ ensures the stabilization of transactions within $\textsf{CTxSet}$ in the ledger of the source chain. Simultaneously, when the committee aggregates the signatures of $\textsf{CTxSet}$, it confirms the transactions as anticipated. Thus, $ACP_A$ commits to the validity of $\textsf{CTxSet}_A$ on organization \textit{A}'s blockchain. Consequently, organization \textit{B} can process corresponding cross-domain transactions on its blockchain whenever $ACP_A$ is received, provided the verification is successful. Thus, DAS enables the continuous generation of ACPs in a discontinuous network and supports rapid public verification of these ACPs.

As described earlier, the leader of organization \textit{A} obtains a validity commitment $(\sigma_{\textsf{DAS}}, \pi_{\textsf{DAS}}, \textsf{CTxSet}_A')$ by running $\texttt{DASig}$ (Line 5). If the committee's aggregate public key $apk_A$ check passes (Line 6), confirming that $ACP_A$ is indeed produced by the honest committee, the leader forwards the DAS commitment and the pre-existing information in an ACP message to organization \textit{B} (Line 7). After receiving $ACP_A$, any node of organization \textit{B}, typically the leader, can verify the validity of the DAS commitment by running $\texttt{DASVer}$ (Line 8). Essentially, $\texttt{DASVer}$ involves checking both $\sigma_{\textsf{DAS}}$ and $\pi_{\textsf{DAS}}$. If both verifications are successful, the authenticity of the cross-domain data $\{d_i^A\}|_{i=1}^{n}$ within $\textsf{CTxSet}_A$ and the associated delay $\delta$ are confirmed. Consequently, this data set is certified as stable and tamper-proof on organization \textit{A}'s blockchain.

Finally, organization \textit{B} performs a sequential check of the transactions in $\textsf{CTxSet}_A$ (Lines 9-10). If the check does not conflict with the order maintained by the source chain, each corresponding cross-domain data $d_i^A$ is recorded onto the blockchain. Once the transactions stabilize on the target blockchain, the cross-domain data for this epoch is transferred. Similarly, the process of transferring data from organization \textit{B} to organization \textit{A} follows the same procedure (Line 11).

\section{System Details of AsyncSC} \label{Building}
In this section, we further describe the asynchronous transaction sequence guarantee mechanism and discuss the details of the delay settings of the DAS.

\subsection{Transaction Sequence Guarantee Mechanism} \label{Buffer}
To maintain the consistency of asynchronous \textsf{CTx}s, as shown in Fig. \ref{4-p}, we propose a restricted-readable transaction sequence guarantee mechanism. This mechanism can be implemented in various ways; here, we provide a buffer pool-based implementation idea. The analysis and optimization of its utility are left for future work.

The sequential check consists of five main steps, as shown in Fig. \ref{4-pool1}. \ding{202} \textit{Receiving}. When the \textsf{MC}'s committee receives the $\textsf{CTxSet}$ for each epoch, its leader maintains these \textsf{CTx}s in the local multilevel buffer pool. To improve efficiency and minimize conflicts, we draw inspiration from Ref. \cite{warner1995version}. The multilevel buffer pool uses an adaptive heap structure to categorize and store transactions according to different priorities (\textit{e.g.}, timestamp, importance, dependency). The adaptive heap dynamically adjusts the order based on the current load situation and the dependencies of the \textsf{CTx}s, ensuring that \textsf{CTx}s with high priority and early timestamps are located at the top of the heap. \ding{203} \textit{Ordering}. The \textsf{MC} leader dynamically adjusts the order of the pooled \textsf{CTx}s through a sliding window. It also dynamically adjusts the window size and movement speed based on real-time transaction traffic. \ding{204} \textit{Processing}. The leader of \textsf{SC} performs conflict checking on the \textsf{CTx}s in each $\textsf{CTxSet}$ received asynchronously based on the buffer pool maintained by \textsf{MC}, and then sends the processed set of \textsf{CTx}s to the committee for confirmation. \ding{205} \textit{Confirming}. The committee of \textsf{SC} uses smart contracts to automate the confirmation of the transaction order and ensure consistency through BFT consensus algorithms such as Tendermint \cite{buchman2016tendermint}. After confirming that there are no conflicts, the relevant \textsf{CTx}s are executed in an orderly manner. \ding{206} \textit{Removing}. Finally, after confirming that the \textsf{CTx}s have stabilized through a persistent query on the ledger, the leader of \textsf{MC} removes the executed \textsf{CTx}s from the buffer pool.

\begin{figure}[t]
	\vspace{-0.2cm}
	\subfigure[Checking steps.]{
		\begin{minipage}[t]{0.231\linewidth}
			\centering
			\hspace*{-3.5mm}
			\includegraphics[width=0.87in]{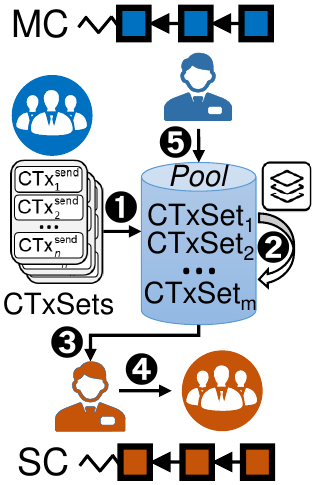}
			\label{4-pool1}
			\vspace{-0.4cm}
		\end{minipage}
	}
	\subfigure[\textsf{CTx} sequence restricted readable.]{
		\begin{minipage}[t]{0.7\linewidth}
			\centering
			\hspace*{-0.35cm}
			\includegraphics[width=2.6in]{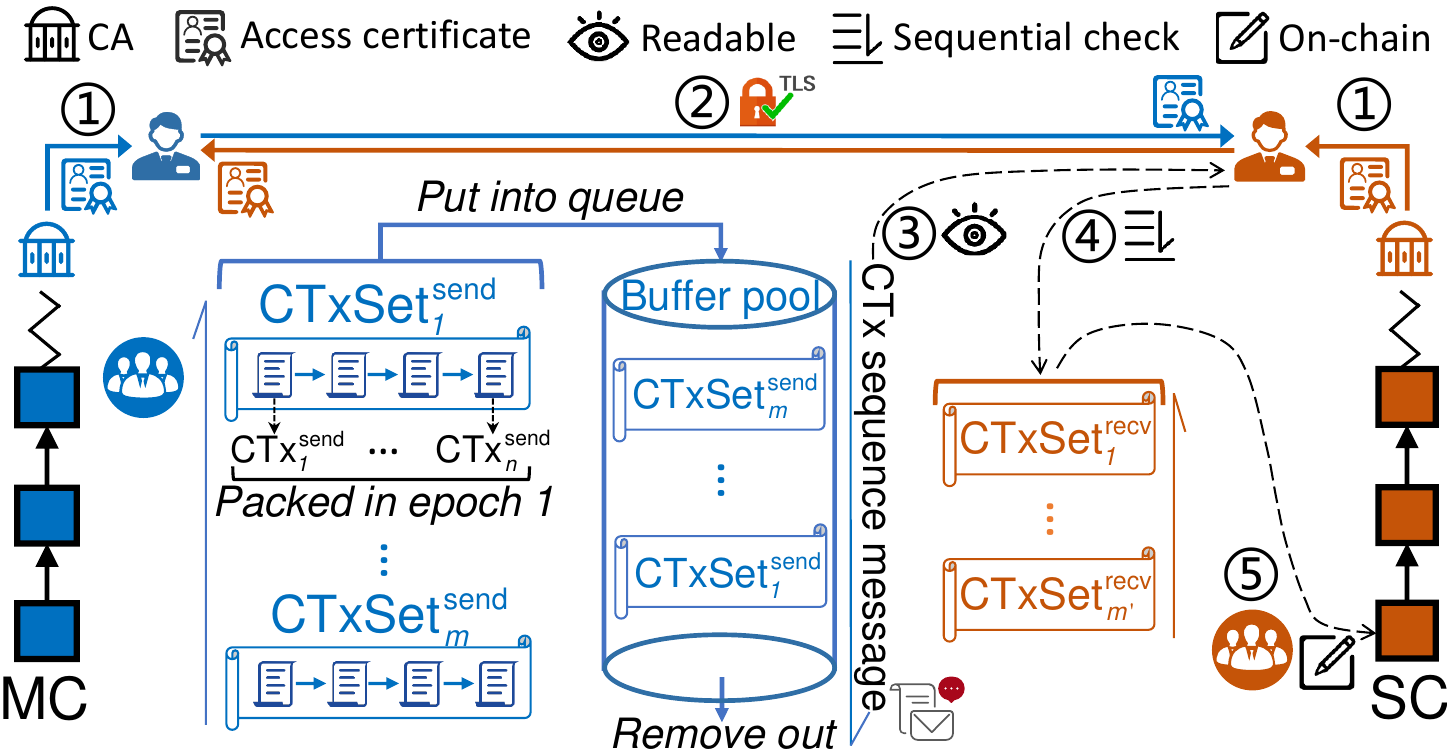}
			\label{4-pool}
			\vspace{-0.4cm}
		\end{minipage}
	}
	\centering
	\vspace{-0.3cm}
	\captionsetup{font=footnotesize, labelsep=period}
	\caption{Asynchronous transaction sequence guarantee mechanism.}
	\label{4-p}
	\vspace{-0.65cm}
\end{figure}

In the context of this paper, to adhere to the internal confidentiality of permissioned blockchains, the \textsf{CTx} sequence maintained by the buffer pool is restricted to be readable only by authenticated leaders, as shown in Fig. \ref{4-pool}. \ding{192} Upon leader nodes' access to the permissioned blockchains, certificate authorities (CAs) issue access certificates to them. \ding{193} When cross-chain interactions are required, leaders on both sides exchange certificates over a secure communication channel such as TLS to authorize each other. \ding{194} Then, the leader of \textsf{SC} can receive a signed message from \textsf{MC} containing only the transaction sequence in the buffer pool. \ding{195} The leader of \textsf{SC} then performs a sequential check on the received \textsf{CTx}s based on this message. \ding{196} The committee of \textsf{SC} completes the on-chain recording of \textsf{CTx}s after checking.

\subsection{Delay Settings for DAS}  \label{DelaySet}
The core idea of implementing DAS lies in the fact that the $\texttt{DASig}$ algorithm produces a controlled delay by performing a series of sequential computations on the input for a given number of steps, such as the iterative computation of the SNARK utilized in the IVC adopted in this paper. At the same time, it ensures that the results can be verified quickly. Under given hardware and algorithmic conditions, the time required to compute $\texttt{DASig}$ for a given length of input is fixed. This predictability of computation time allows the system designers to set specific delays as needed. We ensure that each computation goes through the same required steps and time by choosing the appropriate delay parameter $t$ in the parameter generation phase of the DAS.

With the workflow of Fig. \ref{3-entity} in mind, we now detail the settings for the size of the delay generated by the DAS. Our goal is to ensure that when \textsf{SC} receives the \textsf{ACP}, the delay $\delta$ generated by running $\texttt{DASig}$ guarantees that the \textsf{CTxSet}, \textit{i.e.}, all of the \textsf{CTx}s in an epoch, have stabilized on \textsf{MC}. Thus, the key is to make the last transaction of this epoch on-chain stable. We now discuss the size of the ideal delay $\delta$. These \textsf{CTx}s are recorded on \textsf{MC} in Fig. \ref{3-entity}-\ding{202}. Assume that the last \textsf{CTx} on-chain in this epoch is of length $T_{LastCTx}$ from the last time slot of the epoch. Thus, $T_{LastCTx}$ is included in the duration required for the stabilization of \textsf{CTxSet}. Additionally, the $\texttt{DASig}$ algorithm is started in Fig. \ref{3-entity}-\ding{204}, \textit{i.e.}, the time for the committee to carry out the signatures in Fig. \ref{3-entity}-\ding{203}, $T_{Sig}$, is already accounted for in the length of time required for the stabilization of \textsf{CTxSet}. Let the time required for a transaction to stabilize on the permissioned blockchain be $T_{Stab}$, then the ideal delay $\delta = T_{Stab} - T_{LastCTx} - T_{Sig}$. Since $T_{Stab}$ and $T_{Sig}$ are essentially constant in a specific permissioned blockchain and $T_{LastCTx}$ is easily obtained, it is feasible to set a suitable delay parameter $t$ to control $\delta$.

In practice, the delay $\delta$ of DAS is a tradeoff between security and efficiency, and we usually set $\delta$ slightly larger than its ideal duration to ensure stronger security. The optimization of this tradeoff is left for future work. Additionally, if the transaction volume on \textsf{MC} surges or the intra-chain communication deteriorates, the block may take longer to stabilize, thus requiring a longer $\delta$. For different blockchains, $\delta$ is adjusted according to the duration of their ledger state evolution \cite{zamyatin2021sok}.

%% file: 5-Security.tex
\section{Security Analysis} \label{Security}
We define and prove the security of AsyncSC in this section.

Persistence and liveness are key properties that ensure blockchain security \cite{garay2015bitcoin}. In short, persistence ensures all honest participants agree on stable transactions and their order in the ledger. Liveness ensures transactions from any honest participant are eventually recorded in the ledger of all honest participants, guaranteeing continued ledger progress. Combined with the definition of correct cross-chain communication by Zamyatin \textit{et al.} \cite{zamyatin2021sok}, we define the security of AsyncSC as follows. 

\begin{definition} \label{def-Security}
	{\rm \textbf{(Security).}} For \textsf{MC} and \textsf{SC} with persistence and liveness, AsyncSC must have the following properties:
	
	- {\rm \textbf{Correctness.}} If both $\textsf{CTx}^{\textsf{send}}$ and $\textsf{CTx}^{\textsf{recv}}$ are as expected, they will be recorded in the ledgers $\textsf{L}_{\textsf{MC}}$ and $\textsf{L}_{\textsf{SC}}$, respectively. If either does not match, both parties will abort. Additionally, there is no case where $\textsf{CTx}^{\textsf{send}}$ or $\textsf{CTx}^{\textsf{recv}}$ is recorded in the ledger by only one party.
	
	- {\rm \textbf{Stability.}} If $\textsf{CTx}^{\textsf{send}}$ is as expected, then eventually $\textsf{CTx}^{\textsf{send}}$ will be recorded in $\textsf{L}_{\textsf{MC}}$ and $\textsf{CTx}^{\textsf{recv}}$ will be recorded in $\textsf{L}_{\textsf{SC}}$. Given a common prefix parameter $k \in \mathbb{N}$, the blocks containing $\textsf{CTx}^{\textsf{send}}$ and $\textsf{CTx}^{\textsf{recv}}$ are located more than $k$ blocks away from the ends of $\textsf{L}_{\textsf{MC}}$ and $\textsf{L}_{\textsf{SC}}$, respectively.
\end{definition}

The correctness property combines effectiveness and atomicity as described in Ref. \cite{zamyatin2021sok}. It is a safety property that ensures \textsf{CTx}s match the parties' expectations and maintains the consistency of cross-chain interactions. The stability property guarantees the eventual termination of the cross-chain protocol, making it a liveness property.

\begin{theorem}
	Provided that \textsf{MC} and \textsf{SC} satisfy persistence and liveness properties, AsyncSC is considered secure.
\end{theorem}

\begin{proof}
	By Definition \ref{def-Security}, an arbitrary probabilistic polynomial time (PPT) adversary $\mathcal{A}$ exists in two cases if it violates the correctness property. The notation can be expressed as:
	
	$\bullet$ \textbf{C1:} $\textsf{CTx}^{\textsf{send}} \in \textsf{L}_{\textsf{MC}} \wedge \textsf{CTx}^{\textsf{recv}} \notin \textsf{L}_{\textsf{SC}}$. 
	
	$\bullet$ \textbf{C2:} $\textsf{CTx}^{\textsf{send}} \notin \textsf{L}_{\textsf{MC}} \wedge \textsf{CTx}^{\textsf{recv}} \in \textsf{L}_{\textsf{SC}}$.
	
	For \textbf{C1}, the possible reason is that $\mathcal{A}$ delays/denies sending the ACP to \textsf{SC}. As we mentioned, \textsf{MC} will retransmit \textsf{CTx}s that exceed the threshold $\Delta_{async}$, so the final state will be $\textsf{CTx}^{\textsf{send}} \notin \textsf{L}_{\textsf{MC}} \wedge \textsf{CTx}^{\textsf{recv}} \notin \textsf{L}_{\textsf{SC}}$. If the committee of \textsf{MC} forges signatures, then it needs to corrupt at least $\frac{|S^{C}|}{2} + 1$ members. This has a negligible probability of happening in an honest-majority permissioned blockchain, and violates the security of committee elections \cite{huang2023scheduling, zhai2024secret, boehmer2024approval}.
	
	For \textbf{C2}, the possible reason is that the organization of \textsf{SC} forges signatures. Similarly, $\mathcal{A}$ needs to corrupt at least $\frac{|S^{Org}|}{2}+1$ nodes, which has negligible probability. If $\mathcal{A}$ corrupts the leader of \textsf{MC} to forge ACP and make \textsf{SC} include $\textsf{CTx}^{\textsf{recv}}$ in $\textsf{L}_{\textsf{MC}}$, then it needs to forge the \textsf{MC} committee's aggregate public key and the DAS output to pass $\texttt{DASVer}$. This obviously cannot hold, since a fake $apk$ cannot pass $\texttt{AKCheck}$. Furthermore, the security of DAS relies on the collision-resistant hash function \cite{boldyreva2002efficient}, and the soundness and correctness of the IVC scheme \cite{boneh2018verifiable}. The probability that a PPT adversary $\mathcal{A}$ breaks through a publicly setup DAS is negligible.

	Thus, the probability that $\mathcal{A}$ successfully violates the correctness property is negligible.
	
	We further analyze whether $\mathcal{A}$ can violate the stability property. Suppose there is a valid $\textsf{CTx}^{\textsf{send}}$. Since \textsf{MC} satisfies the liveness property, it will eventually stabilize in $\textsf{L}_{\textsf{MC}}$, \textit{i.e.}, $\textsf{CTx}^{\textsf{send}} \in \textsf{L}_{\textsf{MC}}$. During the protocol run, $\textsf{CTx}^{\textsf{send}}$ will first be packed into the \textsf{CTxSet}. Then, the committee of \textsf{MC} signs the \textsf{CTxSet}. The leader of \textsf{MC} generates the ACP and sends it to \textsf{SC}. After the leader of \textsf{SC} verifies successfully, a sequential check of the \textsf{CTx}s in \textsf{CTxSet} will be performed. Since \textsf{MC} satisfies persistence, $\textsf{CTx}^{\textsf{recv}}$ can maintain the same order as $\textsf{CTx}^{\textsf{send}}$. Thus, valid $\textsf{CTx}^{\textsf{recv}}$ will be recorded in $\textsf{L}_{\textsf{SC}}$. Since \textsf{SC} satisfies liveness, the final state is $\textsf{CTx}^{\textsf{recv}} \in \textsf{L}_{\textsf{SC}}$. Thus, the stability of AsyncSC follows from this.
	
	Therefore, AsyncSC satisfies correctness and stability, proving its security.
\end{proof}	

%% file: 6-PerformanceAnalysis.tex
\begin{figure*}[t]
	\vspace{-0.3cm}
	\centering
	\hspace{-0.8cm}
	\subfigure[Throughput vs. the \# of \textsf{CTx}s.]{
		\begin{minipage}[t]{0.23\linewidth}
			\centering
			\includegraphics[width=1.735in]{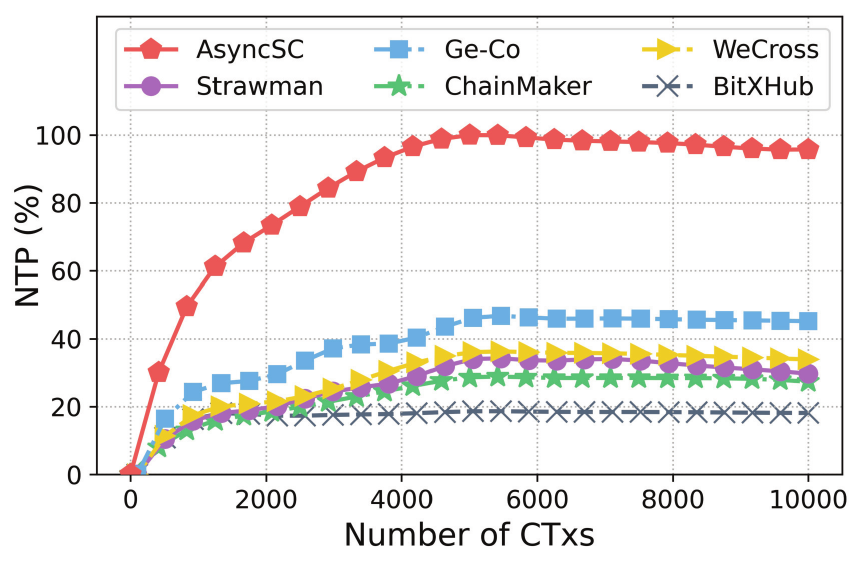}
			\label{TPS-6schemes}
			\vspace{-0.8cm}
		\end{minipage}
	}
	\hspace{-0.2cm}
	\subfigure[Latency vs. the \# of \textsf{CTx}s.]{
		\begin{minipage}[t]{0.23\linewidth}
			\centering
			\includegraphics[width=1.705in]{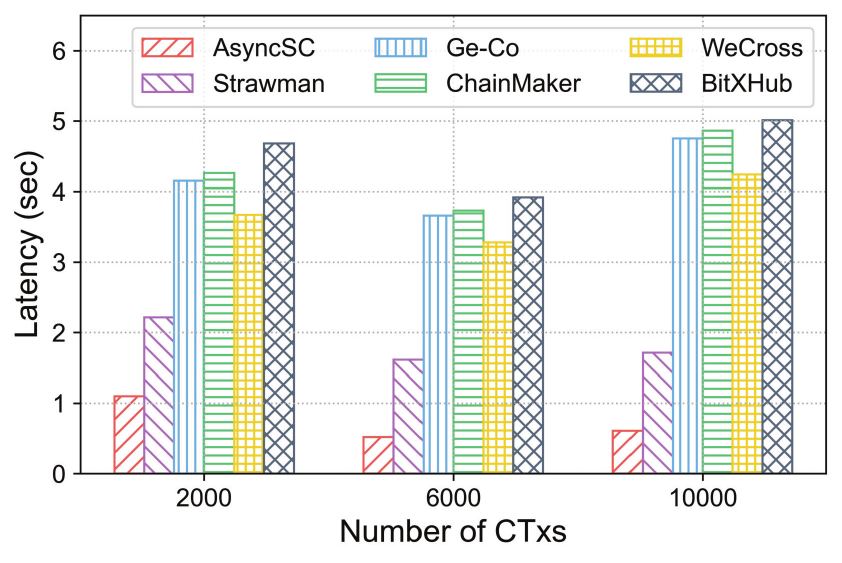}
			\label{Latency-6schemes}
			\vspace{-0.8cm}
		\end{minipage}%
	}%
	\hspace{-0.01cm}
	\subfigure[Success ratio vs. the timeout ratio.]{
		\begin{minipage}[t]{0.23\linewidth}
			\centering
			\includegraphics[width=1.725in]{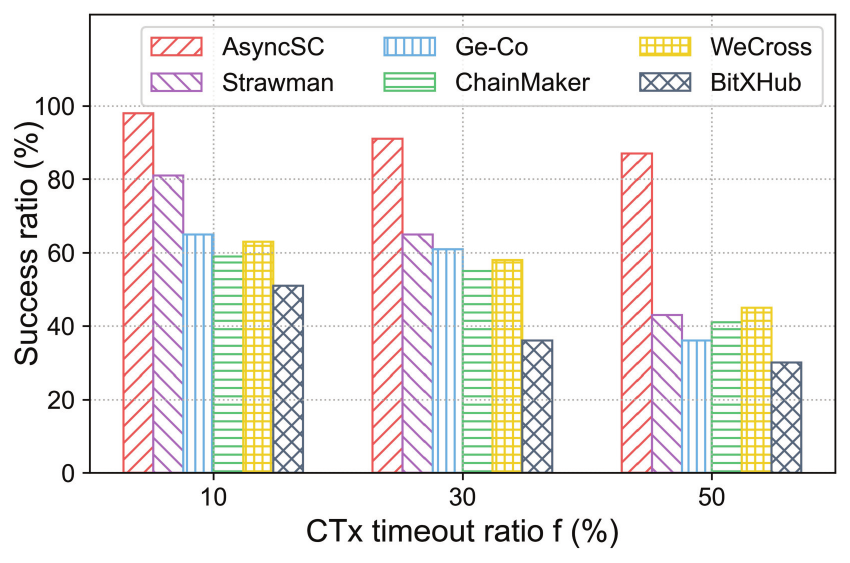}
			\label{SuccRatio-6schemes}
			\vspace{-0.8cm}
		\end{minipage}
	}%
	\hspace{-0.1cm}
	\subfigure[Overhead vs. the \# of \textsf{CTx}s.]{
		\begin{minipage}[t]{0.23\linewidth}
			\centering
			\includegraphics[width=1.95in]{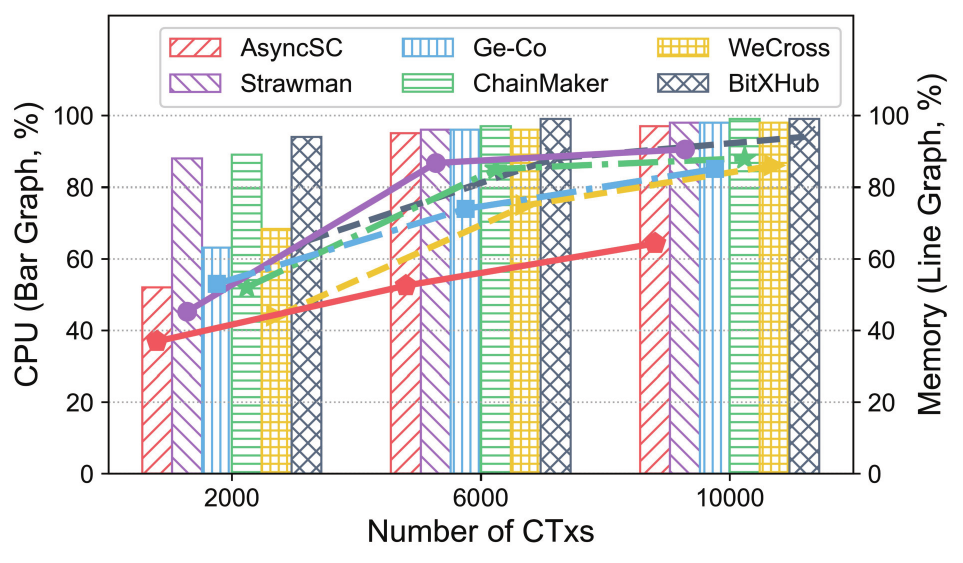}
			\label{Source-6schemes}
			\vspace{-0.8cm}
		\end{minipage}
	}%
	\centering
	\vspace{-0.3cm}
	\captionsetup{font=footnotesize, labelsep=period}
	\caption{Performance comparison of different baselines across various metrics.}
	\label{performance1}
	\vspace{-0.65cm}
\end{figure*}

\section{Performance Evaluation} \label{Evaluation}
\subsection{Implementation} \label{Implementation}
To evaluate AsyncSC, we implement a prototype based on two well-known permissioned blockchains: Hyperledger Fabric \cite{androulaki2018hyperledger} and ChainMaker \cite{chainmaker}, using the \texttt{Go}. We realize DAS by modifying the multi-signature \texttt{MuSig2}\footnote{https://github.com/aureleoules/musig2-coordinator} and Harmony's VDF\footnote{https://github.com/harmony-one/vdf}. The prototype code will be made public soon.

\textbf{Baselines:} For a fair comparison, the baselines for evaluating AsyncSC include two cross-chain prototypes we developed and three enterprise-level projects supporting cross-chain interactions on permissioned blockchains. Their names and core ideas are as follows: (i) \textit{Strawman}. An incomplete AsyncSC prototype using SPV proof for each \textsf{CTx} as a commitment. (ii) \textit{Ge-Co} \cite{yin2023sidechains}. A sidechain scheme that uses committee voting for threshold signatures. (iii) \textit{ChainMaker} \cite{chainmaker}. Uses cross-chain proxy, SPV, and transaction contracts, supporting cross-chain operations with Fabric. (iv) \textit{WeCross} \cite{WeCross}. Relay-based cross-chain routing that enables data interoperability between Fabric and FISCO BCOS \cite{FISCO}. (v) \textit{BitXHub} \cite{shao2020bitxhub}. Provides an InterBlockchain Transfer Protocol (IBTP) based on relay chains for reliable routing between permissioned blockchains.

\subsection{Settings}
\textbf{Testbed:} It consists of 16 virtual machines, each equipped with an Intel i7-13700 CPU, 8 GB of RAM, and a 5 Mbps network link. To simulate semi-synchronous/asynchronous communication, we use the \texttt{tc}\footnote{https://linux.die.net/man/8/tc} command on each machine to control transmission delay $\Delta$ (100 $\sim$ 500 ms) and network jitter ($\pm$ 25 ms), and to generate \textsf{CTx} timeouts (\textit{i.e.}, $\Delta > \Delta_{async}$) with a ratio of $f$. We launch 20 nodes on each machine using Docker containers. We run Hyperledger Fabric v2.4 on half of the machines and ChainMaker v2.1 on the other half to simulate cross-domain data exchanges using heterogeneous permissioned blockchains by different domains. Additionally, the versions of the compared projects are WeCross v1.2 and BitXHub v1.11.

\textbf{Dataset:} We use the real-world IoT dataset \texttt{TON\_IoT}\footnote{https://research.unsw.edu.au/projects/toniot-datasets}, which contains over two million pieces of raw telemetry data collected by weather and Modbus sensors \cite{booij2021ton_iot}. Blockchains from different organizations first store different data and then exchange a number of data per epoch.

\textbf{Metrics:} The following metrics are used to evaluate the performance of cross-chain interactions. (i) \textit{Throughput}. The number of successful executions of \textsf{CTx} per second by the system (TPS). (ii) \textit{Latency}. The duration for a \textsf{CTx} to stabilize on \textsf{SC} from the time it is initiated by \textsf{MC}. (iii) \textit{Success ratio}. The percentage of successful executions in initiated \textsf{CTx}s (TSR). (iv) \textit{CPU and memory utilization}. The resource consumption of the system during execution.

\subsection{Performance Comparison}
We perform comprehensive \textsf{CTx}-driven simulations by adjusting the test parameters. The number of \textsf{CTx} per test ranges from 1,000 to 10,000. We control $\Delta$ to follow a normal distribution between 100 and 500 ms to simulate the typical semi-synchronous communication between IoT devices. At the same time, we control some of the \textsf{CTx}s to timeout (we set $\Delta_{async}=10$ s) proportionally from 10\% to 50\% to simulate asynchronous communication. We set the epoch duration to 600 ms. The block generation time for each blockchain is 100 ms. A block is considered stabilized if it is confirmed by the next five blocks (\textit{i.e.}, $T_{Stab}=500$ ms). The detailed experimental results are shown in Fig. \ref{performance1}.

\subsubsection{Throughput} We first normalize the throughput results (NTP) to eliminate environmental differences between the different schemes and facilitate comparisons. The value of NTP reflects the efficiency of the system in highly concurrent data transfer. Fig. \ref{TPS-6schemes} shows the curves of throughput with an increasing number of \textsf{CTx}. Initially, the NTP of all schemes rises with the number of \textsf{CTx}s. The performance bottleneck is approached when the number of \textsf{CTx}s reaches 4,000 to 6,000. However, AsyncSC shows an average increase in throughput of 1.21 $\sim$ 3.96 $\times$ compared to other schemes. The reason for AsyncSC's advantage in throughput is twofold. First, AsyncSC improves proof efficiency by batch committing packaged \textsf{CTx}s through ACP. Second, the test simulates asynchronous communication with timed-out \textsf{CTx}s, where traditional synchronous schemes experience more transaction aborts. In contrast, AsyncSC demonstrates the benefits of asynchronous concurrency. Additionally, the actual throughput and our improvements for all schemes are shown in Table II.

\begin{table}[ht]
	\vspace{-0.3cm}
	\centering
	\captionsetup{font=footnotesize}
	\caption{\textsc{Throughput of AsyncSC and Baselines.}}
	\vspace{-0.2cm}
	\resizebox{1\columnwidth}{!}{
	\begin{tabular}{lcccccc}
		\toprule
		Schemes  & Strawman & Ge-Co & ChainMaker & WeCross & BitXHub & \textbf{AsyncSC} \\
		\midrule
		Avg. TPS & 457.93 & 652.70 & 402.92 & 497.40 & 290.82 & \textbf{1442.47} \\
		Avg. increase & 2.15 $\times$ & 1.21 $\times$ & 2.58 $\times$ & 1.90 $\times$ & 3.96 $\times$ & - \\
		\bottomrule
	\end{tabular}
	}
	\label{TPS-6s}
	\vspace{-0.2cm}
\end{table}

\subsubsection{Latency} Fig. \ref{Latency-6schemes} shows the average latency from initiation to stabilization for each \textsf{CTx} at transaction numbers of 2,000, 6,000, and 10,000. The results show that AsyncSC reduces the latency by 59.76\% to 83.61\% compared to the other schemes. Notably, the average latency of AsyncSC is 0.74 s. Although there exists a message delay, AsyncSC uses ACP to commit an epoch of packed \textsf{CTx}s. Even if part of the transaction times out and retransmits, the asynchronous concurrency greatly reduces the overall transaction latency. AsyncSC's advantages in asynchronous large-scale data exchange scenarios will be more pronounced.

\subsubsection{Success ratio} Fig. \ref{SuccRatio-6schemes} shows the average success ratio of 10,000 \textsf{CTx}s tested with timeout ratio $f$ set to 10\%, 30\%, and 50\%. The results show that AsyncSC increases the TSR by 46.03\% to 135.9\% compared to the other schemes. Furthermore, the relative improvement of AsyncSC is more pronounced as $f$ increases. We find that schemes under traditional synchronous settings struggle to handle \textsf{CTx} delays or timeouts properly. Even if they can, transactions are usually aborted due to conflicts. However, AsyncSC does not require \textsf{MC} to continuously monitor the state of \textsf{CTx}s, and its asynchronous transaction sequence guarantee mechanism can effectively reduce the occurrence of transaction conflicts.

\subsubsection{CPU and memory utilization} Fig. \ref{Source-6schemes} shows the main results of the different schemes in terms of resource overhead during data exchange. It can be seen that the systems are almost fully loaded when the number of \textsf{CTx}s reaches 6,000, and thus the performance reaches a bottleneck. However, AsyncSC reduces the overall CPU and memory utilization by an average of 5.1\% to 16.44\% and 25.02\% to 37.77\%, respectively. This is due to the fact that AsyncSC uses ACP for batch \textsf{CTx} commitment and reduces the overhead of monitoring and synchronizing states between blockchains.

\begin{table}[t]
	\vspace{-0.2cm}
	\centering
	\captionsetup{font=footnotesize}
	\caption{\textsc{Comparison of Cross-Chain Proofs.}}
	\vspace{-0.2cm}
	\resizebox{1\columnwidth}{!}{
		\begin{threeparttable}
			\begin{tabular}{lcccc}
				\toprule
				Schemes & PoW sidechains \cite{kiayias2019proof} & PoS sidechains \cite{gavzi2019proof} & Ge-Co \cite{yin2023sidechains} & \textbf{AsyncSC} \\
				\midrule
				Proof size$^1$ & $\mathcal{O}(\log cl)$ & $\mathcal{O}(k)$ & $\mathcal{O}(1)$ & $\mathcal{O}(1)$ \\
				Comp. cost$^2$ & $\mathcal{O}(\log cl)$ & $\mathcal{O}(k)$ & $\mathcal{O}(|S^{C}|)$ & $\mathcal{O}(\frac{|S^{C}|}{n})$ \\
				\midrule
				Gen. time$^3$ & 1205.12 s & 54.54 s & 3.91 s & \textbf{512.3 ms} \\	
				Ver. time & 30.97 s & 24.72 s & 618.37 ms & \textbf{83.6 ms} \\		
				\bottomrule
			\end{tabular}
			\begin{tablenotes}
				\item[$1$] The symbol $cl$ denotes the chain length, and $k$ is the common prefix parameter.
				\item[$2$] The $|S^{C}|$ denotes committee size, and $n$ denotes the number of \textsf{CTx}s in an epoch.
				\item[$3$] Measure the time to generate cross-chain proofs for an epoch containing 1,200 \textsf{CTx}s. We set $cl \approx 2.03 \times 10^7$ (as of July 2024 on ETH); $k=2160$; and a committee size of 10.
			\end{tablenotes}
		\end{threeparttable}
	}
	\label{Proofs-6s}
	\vspace{-0.6cm}
\end{table}

\subsection{Evaluation of ACP}
Table \ref{Proofs-6s} provides a theoretical comparison between AsyncSC and some SOTA sidechain constructions. AsyncSC is based on the DAS realization of generating ACPs for $n$ \textsf{CTx}s within an epoch. The cross-chain proofs are aggregated from the committee's signatures. In addition, delayed proofs are implemented in DAS using IVC iterative computations. For any sub-exponential length computation, the complexity in terms of proof size and verification cost is limited by \textsf{poly}$(\lambda)$. This part of the proof size and computational cost is negligible \cite{boneh2018verifiable}. Its size is not affected by the chain length or the common prefix. Thus, the proof size complexity of AsyncSC is $\mathcal{O}(1)$. The computational cost shared by a single \textsf{CTx} is $\mathcal{O}(\frac{|S^{C}|}{n})$. AsyncSC is more lightweight than the SOTA constructions.

We measure the generation and verification times for cross-chain proofs. The results in Table \ref{Proofs-6s} show that the average generation and verification times of AsyncSC are 512.3 ms and 83.6 ms, respectively. Since AsyncSC needs to generate only one cross-chain proof for the \textsf{CTx}s of an epoch, which can be verified quickly, it reduces the proof generation and verification times by 1 to 4 and 1 to 3 orders of magnitude, respectively, compared to SOTA schemes.

Next, we vary the size of $|S^{C}|$ and $n$ to evaluate the system throughput versus DAS delay of ACP, as shown in Fig. \ref{DAS-2}. We vary the DAS delay within the range $\delta = 400 \sim 600$ ms. Since $\Delta$ is a minimum of 100 ms and $T_{Stab}=500$ ms, $\delta$ ensures \textsf{CTx} stability. The results show that throughput increases and then decreases with the increase of $\delta$. This is because the size of $\delta$ affects the number of \textsf{CTx} arriving at an epoch, which in turn affects the adequacy of ACP aggregation. Fig. \ref{TPSvsDelay1} shows the results of varying $|S^{C}|$ for a fixed \textsf{CTx} arrival rate (AR) of 2,000 per second. It can be seen that the larger $|S^{C}|$ is, the lower the throughput, but its effect is not significant. This is because the larger $|S^{C}|$, the more committee signatures need to be aggregated, and the aggregation process is rapid. In addition, we control the number of \textsf{CTx} per epoch (\textit{i.e.}, $n$) by varying the AR. Fig. \ref{TPSvsDelay2} shows the effect of different ARs on throughput. An AR that is too small or too large results in too few or too many \textsf{CTx}s being processed at each epoch, potentially causing the system to be underloaded or overloaded, thereby affecting system throughput.

\begin{figure}[t]
	\vspace{-0.5cm}
	\subfigure[\textsf{CTx} arrival rate = 2000 \textsf{CTx}/sec.]{
		\begin{minipage}[t]{0.45\linewidth}
			\centering
			\hspace*{-3mm}
			\includegraphics[width=1.75in]{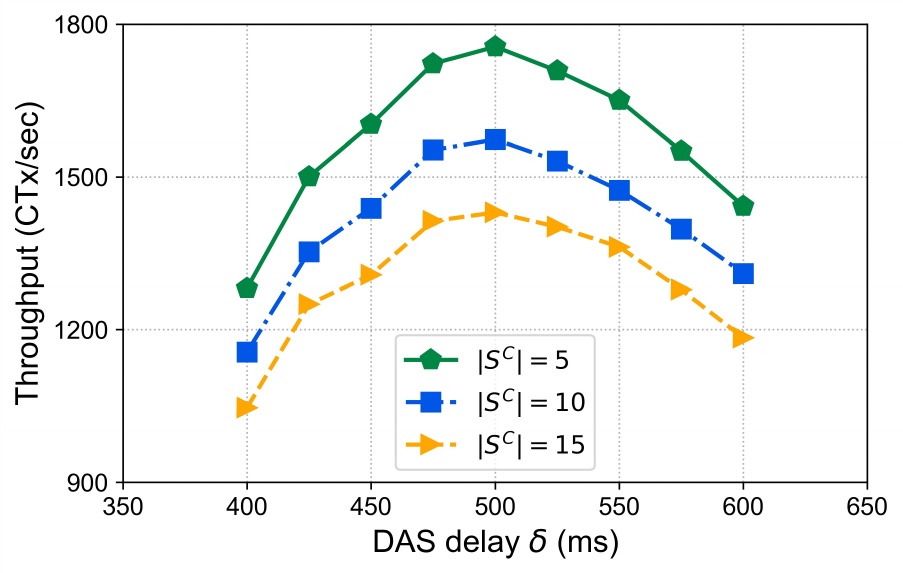}
			\label{TPSvsDelay1}
			\vspace{-0.4cm}
		\end{minipage}
	}
	\subfigure[Committee size $|S^{C}|$ = 5.]{
		\begin{minipage}[t]{0.45\linewidth}
			\centering
			\hspace*{-1.5mm}
			\includegraphics[width=1.75in]{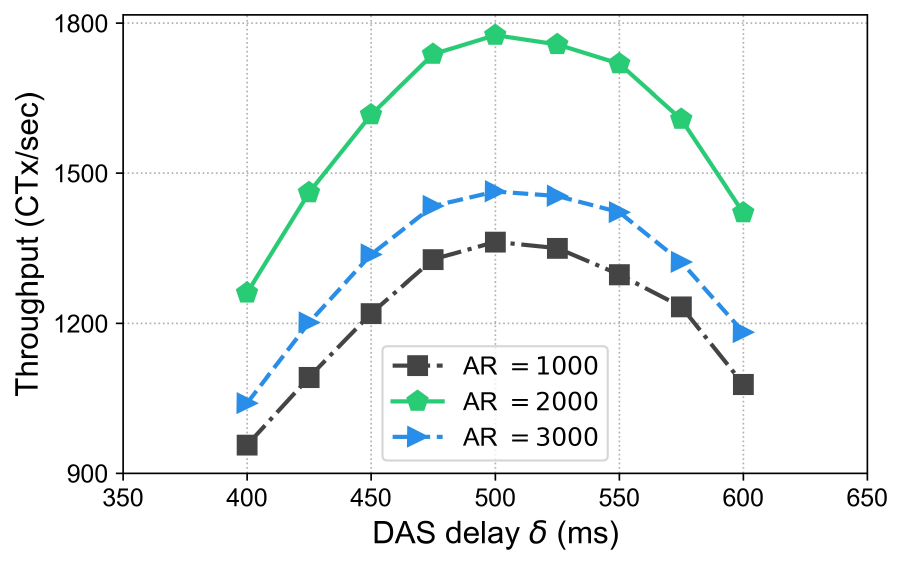}
			\label{TPSvsDelay2}
			\vspace{-0.4cm}
		\end{minipage}
	}
	\centering
	\vspace{-0.3cm}
	\captionsetup{font=footnotesize, labelsep=period}
	\caption{Evaluating system throughput vs. ACP's DAS delay.}
	\label{DAS-2}
	\vspace{-0.65cm}
\end{figure}

\begin{table}[h]
	\vspace{-0.3cm}
	\centering
	\captionsetup{font=footnotesize}
	\caption{\textsc{Evaluating of \textsf{CTx} Buffer Pool.}}\vspace{-0.2cm}
	\resizebox{1\columnwidth}{!}{
		\begin{tabular}{cccccccccc}
			\toprule
			Metrics & \multicolumn{3}{c}{Queue size (\# of \textsf{CTx})} & \multicolumn{3}{c}{Window size (\# of \textsf{CTx})} & \multicolumn{3}{c}{Heap priority} \\
			\midrule
			Options & 1000 & \textbf{2000} & 3000 & 200 & \textbf{400} & 600 & Time. & Imp. & \textbf{Dep.} \\
			\midrule
			TSR (\%)  & 90.45 & \textbf{92.31} & 91.17 & 91.36 & \textbf{92.52} & 91.02 & 92.12 & 90.87 & \textbf{94.45} \\
			\bottomrule
		\end{tabular}%
	}
	\label{Eval-poll}%
	\vspace{-0.6cm}
\end{table}%

\subsection{Evaluation of \textsf{CTx} buffer pool}
To verify the effectiveness of the transaction sequence guarantee mechanism in reducing conflicts, we evaluate the impact of different metric options for the buffer pool on the transaction success ratio. We attach random weights and 20\% dependencies to the transactions. The timeout rate $f=30\%$. Table \ref{Eval-poll} presents the average results of our multiple control variable tests. Each metric displays three key options. The results indicate that moderately increasing the queue size can improve TSR, but an excessively large queue may reduce processing efficiency. A larger window helps improve TSR, but an excessively large window may increase the probability of conflicts. The dependency prioritization strategy improves TSR more than timestamp and importance prioritization, indicating that considering inter-transaction dependencies in multilevel buffer pools is effective.

%% file: 10-Conclusion.tex
\section{Conclusion} \label{conclusion}
To facilitate the trusted exchange of IoT multi-domain data, we propose AsyncSC, an efficient asynchronous sidechain construction. We present Cross-Blockchain as a Service (C-BaaS) based on a committee. To provide asynchronous cross-chain proofs (ACPs), we introduce a novel cryptographic primitive called the delayed aggregate signature (DAS). To ensure the consistency of asynchronous cross-chain transactions, we design a transaction sequence guarantee mechanism. Our analysis demonstrates the security of AsyncSC. A prototype evaluation of AsyncSC shows that it outperforms SOTA schemes. Our future work will focus on optimizing the guarantee of asynchronous transaction consistency and DAS delay strategies.

%% file: 0-Infocom-2025.bbl
\begin{thebibliography}{10}
\providecommand{\url}[1]{#1}
\csname url@samestyle\endcsname
\providecommand{\newblock}{\relax}
\providecommand{\bibinfo}[2]{#2}
\providecommand{\BIBentrySTDinterwordspacing}{\spaceskip=0pt\relax}
\providecommand{\BIBentryALTinterwordstretchfactor}{4}
\providecommand{\BIBentryALTinterwordspacing}{\spaceskip=\fontdimen2\font plus
\BIBentryALTinterwordstretchfactor\fontdimen3\font minus
  \fontdimen4\font\relax}
\providecommand{\BIBforeignlanguage}[2]{{%
\expandafter\ifx\csname l@#1\endcsname\relax
\typeout{** WARNING: IEEEtran.bst: No hyphenation pattern has been}%
\typeout{** loaded for the language `#1'. Using the pattern for}%
\typeout{** the default language instead.}%
\else
\language=\csname l@#1\endcsname
\fi
#2}}
\providecommand{\BIBdecl}{\relax}
\BIBdecl

\bibitem{he2022collaborative}
S.~He, K.~Shi, C.~Liu, B.~Guo, J.~Chen, and Z.~Shi, ``Collaborative sensing in
  internet of things: A comprehensive survey,'' \emph{IEEE Communications
  Surveys \& Tutorials}, vol.~24, no.~3, pp. 1435--1474, 2022.

\bibitem{liu2018survey}
J.~Liu, H.~Shen, H.~S. Narman, W.~Chung, and Z.~Lin, ``A survey of mobile
  crowdsensing techniques: A critical component for the internet of things,''
  \emph{ACM Transactions on Cyber-Physical Systems}, vol.~2, no.~3, pp. 1--26,
  2018.

\bibitem{Cisco2023}
\BIBentryALTinterwordspacing
Cisco. Powering an inclusive, digital future for all. [Online]. Available:
  \url{https://newsroom.cisco.com/c/r/newsroom/en/us/a/y2023/m01/powering-an-inclusive-digital-future-for-all.html}
\BIBentrySTDinterwordspacing

\bibitem{liu2023survey}
Y.~Liu, J.~Wang, Z.~Yan, Z.~Wan, and R.~J{\"a}ntti, ``{A survey on
  blockchain-based trust management for Internet of Things},'' \emph{IEEE
  Internet of Things Journal}, vol.~10, no.~7, pp. 5898--5922, 2023.

\bibitem{tong2022blockchain}
W.~Tong, X.~Dong, Y.~Shen, X.~Jiang \emph{et~al.}, ``A blockchain-driven data
  exchange model in multi-domain iot with controllability and parallelity,''
  \emph{Future Generation Computer Systems}, vol. 135, pp. 85--94, 2022.

\bibitem{tong2023ti}
W.~Tong, X.~Dong, Y.~Zhang, Z.~Zhang, L.~Yang, W.~Yang, and Y.~Shen,
  ``{TI-BIoV: Traffic information interaction for blockchain-based IoV with
  trust and incentive},'' \emph{IEEE Internet of Things Journal}, 2023.

\bibitem{dib2018consortium}
O.~Dib, K.-L. Brousmiche, A.~Durand, E.~Thea, and E.~B. Hamida, ``Consortium
  blockchains: Overview, applications and challenges,'' \emph{Int. J. Adv.
  Telecommun}, vol.~11, no.~1, pp. 51--64, 2018.

\bibitem{feng2021consortium}
Y.~Feng, W.~Zhang, X.~Luo, and B.~Zhang, ``{A consortium blockchain-based
  access control framework with dynamic orderer node selection for 5G-enabled
  industrial IoT},'' \emph{IEEE Transactions on Industrial Informatics},
  vol.~18, no.~4, pp. 2840--2848, 2021.

\bibitem{zhang2024time}
M.~Zhang, Q.~Qu, L.~Ning, and J.~Fan, ``On time-aware cross-blockchain data
  migration,'' \emph{Tsinghua Science and Technology}, vol.~29, no.~6, pp.
  1810--1820, 2024.

\bibitem{Back2014EnablingBI}
A.~Back, M.~Corallo, L.~Dashjr, M.~Friedenbach, G.~Maxwell, A.~Miller,
  A.~Poelstra, J.~Tim{\'o}n, and P.~Wuille, ``Enabling blockchain innovations
  with pegged sidechains,'' vol.~72, 2014, pp. 201--224.

\bibitem{kiayias2019proof}
A.~Kiayias and D.~Zindros, ``Proof-of-work sidechains,'' in \emph{International
  Conference on Financial Cryptography and Data Security (FC)}.\hskip 1em plus
  0.5em minus 0.4em\relax Springer, 2019, pp. 21--34.

\bibitem{gavzi2019proof}
P.~Ga{\v{z}}i, A.~Kiayias, and D.~Zindros, ``Proof-of-stake sidechains,'' in
  \emph{2019 IEEE Symposium on Security and Privacy (SP)}, 2019, pp. 139--156.

\bibitem{yin2021sidechains}
L.~Yin, J.~Xu, and Q.~Tang, ``Sidechains with fast cross-chain transfers,''
  \emph{IEEE Transactions on Dependable and Secure Computing}, vol.~19, no.~6,
  pp. 3925--3940, 2021.

\bibitem{kwon2019cosmos}
J.~Kwon and E.~Buchman, ``Cosmos whitepaper,'' \emph{A Netw. Distrib. Ledgers},
  2019.

\bibitem{FISCO}
\BIBentryALTinterwordspacing
{FISCO BCOS}. [Online]. Available: \url{http://www.fisco-bcos.org/}
\BIBentrySTDinterwordspacing

\bibitem{shao2020bitxhub}
S.~Ye, X.~Wang, C.~Xu \emph{et~al.}, ``Bitxhub: side-relay chain based
  heterogeneous blockchain interoperable platform,'' \emph{Comput Sci},
  vol.~47, no.~06, pp. 300--308, 2020.

\bibitem{yin2023sidechains}
L.~Yin, J.~Xu, K.~Liang, and Z.~Zhang, ``Sidechains with optimally succinct
  proof,'' \emph{IEEE Transactions on Dependable and Secure Computing},
  vol.~21, no.~4, pp. 3375--3389, 2024.

\bibitem{garoffolo2020zendoo}
A.~Garoffolo, D.~Kaidalov, and R.~Oliynykov, ``Zendoo: A zk-snark verifiable
  cross-chain transfer protocol enabling decoupled and decentralized
  sidechains,'' in \emph{2020 IEEE 40th International Conference on Distributed
  Computing Systems (ICDCS)}.\hskip 1em plus 0.5em minus 0.4em\relax IEEE,
  2020, pp. 1257--1262.

\bibitem{xie2022zkbridge}
T.~Xie, J.~Zhang, Z.~Cheng, F.~Zhang, Y.~Zhang, Y.~Jia, D.~Boneh, and D.~Song,
  ``zkbridge: Trustless cross-chain bridges made practical,'' in
  \emph{Proceedings of the 2022 ACM SIGSAC Conference on Computer and
  Communications Security}, 2022, pp. 3003--3017.

\bibitem{androulaki2018hyperledger}
E.~Androulaki, A.~Barger, V.~Bortnikov, C.~Cachin, K.~Christidis, A.~De~Caro,
  D.~Enyeart, C.~Ferris \emph{et~al.}, ``Hyperledger fabric: a distributed
  operating system for permissioned blockchains,'' in \emph{Proceedings of the
  thirteenth EuroSys conference}, 2018, pp. 1--15.

\bibitem{chainmaker}
\BIBentryALTinterwordspacing
{ChainMaker}. [Online]. Available: \url{https://chainmaker.org.cn/}
\BIBentrySTDinterwordspacing

\bibitem{david2018ouroboros}
B.~David, P.~Ga{\v{z}}i, A.~Kiayias, and A.~Russell, ``Ouroboros praos: An
  adaptively-secure, semi-synchronous proof-of-stake blockchain,'' in
  \emph{EUROCRYPT 2018: 37th Annual International Conference on the Theory and
  Applications of Cryptographic Techniques}, 2018, pp. 66--98.

\bibitem{muzammal2019renovating}
M.~Muzammal, Q.~Qu, and B.~Nasrulin, ``Renovating blockchain with distributed
  databases: An open source system,'' \emph{Future generation computer
  systems}, vol.~90, pp. 105--117, 2019.

\bibitem{lu2021blockchain}
J.~Lu, J.~Shen, P.~Vijayakumar, and B.~B. Gupta, ``Blockchain-based secure data
  storage protocol for sensors in the industrial internet of things,''
  \emph{IEEE Transactions on Industrial Informatics}, vol.~18, no.~8, pp.
  5422--5431, 2021.

\bibitem{zhou2022fair}
L.~Zhou, A.~Fu, G.~Yang, Y.~Gao, S.~Yu, and R.~H. Deng, ``Fair cloud auditing
  based on blockchain for resource-constrained iot devices,'' \emph{IEEE
  Transactions on Dependable and Secure Computing}, vol.~20, no.~5, pp.
  4325--4342, 2023.

\bibitem{boneh2001short}
D.~Boneh, B.~Lynn, and H.~Shacham, ``Short signatures from the weil pairing,''
  in \emph{International conference on the theory and application of cryptology
  and information security}.\hskip 1em plus 0.5em minus 0.4em\relax Springer,
  2001, pp. 514--532.

\bibitem{zhao2019practical}
Y.~Zhao, ``Practical aggregate signature from general elliptic curves, and
  applications to blockchain,'' in \emph{ACM asia conference on computer and
  communications security}, 2019, pp. 529--538.

\bibitem{boneh2018verifiable}
D.~Boneh, J.~Bonneau, B.~B{\"u}nz, and B.~Fisch, ``Verifiable delay
  functions,'' in \emph{Annual international cryptology conference}.\hskip 1em
  plus 0.5em minus 0.4em\relax Springer, 2018, pp. 757--788.

\bibitem{huang2023scheduling}
H.~Huang, X.~Peng, Y.~Lin, M.~Xu, G.~Ye, Z.~Zheng, and S.~Guo, ``Scheduling
  most valuable committees for the sharded blockchain,'' \emph{IEEE/ACM
  Transactions on Networking}, vol.~31, pp. 3284--3299, 2023.

\bibitem{zhai2024secret}
M.~Zhai, Q.~Wu, Y.~Liu, B.~Qin, X.~Dai, Q.~Gao, and W.~Susilo, ``Secret
  multiple leaders \& committee election with application to sharding
  blockchain,'' \emph{IEEE Transactions on Information Forensics and Security},
  vol.~19, pp. 5060--5074, 2024.

\bibitem{boehmer2024approval}
N.~Boehmer, M.~Brill, A.~Cevallos, J.~Gehrlein, L.~S{\'a}nchez-Fern{\'a}ndez,
  and U.~Schmidt-Kraepelin, ``Approval-based committee voting in practice: a
  case study of (over-) representation in the polkadot blockchain,'' in
  \emph{Proceedings of the AAAI Conference on Artificial Intelligence},
  vol.~38, no.~9, 2024, pp. 9519--9527.

\bibitem{castro1999practical}
M.~Castro, B.~Liskov \emph{et~al.}, ``Practical byzantine fault tolerance,'' in
  \emph{OSDI}, vol.~99, 1999, pp. 173--186.

\bibitem{boneh2003aggregate}
D.~Boneh, C.~Gentry, B.~Lynn, and H.~Shacham, ``Aggregate and verifiably
  encrypted signatures from bilinear maps,'' in \emph{EUROCRYPT 2003:
  International Conference on the Theory and Applications of Cryptographic
  Techniques}, 2003, pp. 416--432.

\bibitem{nick2021musig2}
J.~Nick, T.~Ruffing, and Y.~Seurin, ``Musig2: Simple two-round schnorr
  multi-signatures,'' in \emph{Annual International Cryptology
  Conference}.\hskip 1em plus 0.5em minus 0.4em\relax Springer, 2021, pp.
  189--221.

\bibitem{wesolowski2020efficient}
B.~Wesolowski, ``Efficient verifiable delay functions,'' \emph{Journal of
  Cryptology}, vol.~33, pp. 2113--2147, 2020.

\bibitem{ephraim2020continuous}
N.~Ephraim, C.~Freitag, I.~Komargodski, and R.~Pass, ``Continuous verifiable
  delay functions,'' in \emph{EUROCRYPT 2020: 39th Annual International
  Conference on the Theory and Applications of Cryptographic Techniques}, 2020,
  pp. 125--154.

\bibitem{boldyreva2002efficient}
A.~Boldyreva, ``Efficient threshold signatures, multisignature and blind
  signature schemes based on the gap-diffie-hellman-group signature scheme,''
  in \emph{PKC 2003}, 2002, pp. 31--46.

\bibitem{warner1995version}
A.~Warner and T.~F. Keefe, ``Version pool management in a multilevel secure
  multiversion transaction manager,'' in \emph{Proceedings 1995 IEEE Symposium
  on Security and Privacy}, 1995, pp. 169--182.

\bibitem{buchman2016tendermint}
E.~Buchman, ``Tendermint: Byzantine fault tolerance in the age of
  blockchains,'' Ph.D. dissertation, University of Guelph, 2016.

\bibitem{zamyatin2021sok}
A.~Zamyatin, M.~Al-Bassam, D.~Zindros, E.~Kokoris-Kogias, P.~Moreno-Sanchez,
  A.~Kiayias, and W.~J. Knottenbelt, ``Sok: Communication across distributed
  ledgers,'' in \emph{International Conference on Financial Cryptography and
  Data Security (FC)}.\hskip 1em plus 0.5em minus 0.4em\relax Springer, 2021,
  pp. 3--36.

\bibitem{garay2015bitcoin}
J.~A. Garay, A.~Kiayias, and N.~Leonardos, ``The bitcoin backbone protocol:
  Analysis and applications,'' \emph{Journal of the ACM}, 2015.

\bibitem{WeCross}
\BIBentryALTinterwordspacing
{WeCross}. [Online]. Available:
  \url{https://github.com/WeBankBlockchain/WeCross}
\BIBentrySTDinterwordspacing

\bibitem{booij2021ton_iot}
T.~M. Booij, I.~Chiscop, E.~Meeuwissen, N.~Moustafa \emph{et~al.}, ``{ToN\_IoT:
  The role of heterogeneity and the need for standardization of features and
  attack types in IoT network intrusion data sets},'' \emph{IEEE Internet of
  Things Journal}, vol.~9, no.~1, pp. 485--496, 2021.

\end{thebibliography}
